\documentclass[a4paper,fleqn]{llncs}

\usepackage[utf8]{inputenc}
\usepackage{microtype}
\usepackage{todonotes}
\usepackage{amsmath}
\usepackage{amssymb}

\usepackage{xstring}
\usepackage{url}
\usepackage{tabularx}
\usepackage{xspace}
\usepackage{nicefrac}

\usepackage{tikz}
\usetikzlibrary{decorations.pathreplacing}
\usetikzlibrary{decorations.pathmorphing}
\usetikzlibrary{shapes,shapes.symbols,automata,arrows}
\usetikzlibrary{calc}
\usetikzlibrary{patterns}

\tikzset{p0/.style = {shape = circle, draw, thick, minimum size = 0.7cm}}
\tikzset{p1/.style = {rectangle, minimum size=.7cm, draw, thick}}
\tikzset{>=stealth, shorten >=1pt}
\tikzset{every edge/.style = {thick, ->, draw}}
\tikzset{every loop/.style = {thick, ->, draw}}
\tikzstyle{fault}=[dashed]

\pgfdeclarelayer{background}
\pgfsetlayers{background,main}

\newcommand{\ParityVertexZero}[4]{
	\node[p0,inner sep=1pt] (#1) at #4 {$\nicefrac{#2}{#3}$};
}
\newcommand{\ParityVertexOne}[4]{
	\node[p1] (#1) at #4 {$\nicefrac{#2}{#3}$};
}

\newcommand{\openquestion}[1]{\textbf{#1}}

\makeatletter
\tikzset{circle split part fill/.style  args={#1,#2}{%
 alias=tmp@name, 
  postaction={%
    insert path={
     \pgfextra{%
     \pgfpointdiff{\pgfpointanchor{\pgf@node@name}{center}}%
                  {\pgfpointanchor{\pgf@node@name}{east}}%
     \pgfmathsetmacro\insiderad{\pgf@x}
      \fill[#1] (\pgf@node@name.base) ([xshift=-\pgflinewidth]\pgf@node@name.east) arc
                          (0:180:\insiderad-\pgflinewidth)--cycle;
      \fill[#2] (\pgf@node@name.base) ([xshift=\pgflinewidth]\pgf@node@name.west)  arc
                           (180:360:\insiderad-\pgflinewidth)--cycle; 
         }}}}}  
 \makeatother

\usepackage{booktabs}

\definecolor{myred}{rgb}{1,0.604,0.604}
\definecolor{mydarkred}{rgb}{1,0.345,0.345}
\definecolor{myblue}{rgb}{0.635,0.675,0.966}
\definecolor{mydarkblue}{rgb}{0.412,0.475,0.957}
\definecolor{myyellow}{rgb}{1,0.976,0.604}
\definecolor{mydarkyellow}{rgb}{1,0.961,0.345}

\tikzset{
	assign/.style = { fill=myblue },
	choice/.style = { fill=myred },
	check/.style  = { fill=myyellow }
}

\newcommand{\myquot}[1]{``#1''}

\newcommand{\bigo}[0]{\mathcal{O}}

\newcommand{\size}[1]{|#1|}
\newcommand{\card}[1]{\size{#1}}
\newcommand{\set}[1]{\{ #1 \}}
\renewcommand{\epsilon}{\varepsilon}

\newcommand{\arena}{\mathcal{A}}
\newcommand{\game}{\mathcal{G}}
\newcommand{\col}{\Omega}
\newcommand{\wincond}{\mathrm{Win}}
\newcommand{\winreg}{\mathcal{W}}

\newcommand{\safety}{\mathrm{Safety}}
\newcommand{\parity}{\mathrm{Parity}}

\newcommand{\mem}{\mathcal{M}}
\newcommand{\init}{\mathrm{Init}}
\newcommand{\update}{\mathrm{Upd}}
\newcommand{\nxt}{\mathrm{Nxt}}
\newcommand{\ext}{\mathrm{ext}}

\newcommand{\exptime}{\textsc{ExpTime}}

\newcommand{\dom}{\mathrm{dom}}
\newcommand{\im}{\mathrm{im}}
\newcommand{\disturbances}[0]{\#_{d}}
\newcommand{\cur}{\mathrm{cur}}

\newcommand{\sigmaf}{\sigma_{\!f}}
\newcommand{\sigmaomega}{\sigma_{\!\omega}}
\newcommand{\rig}{\mathrm{rig}}

\usepackage{fullpage}
\usepackage{hyperref}

\title{Synthesizing Optimally Resilient Controllers\thanks{Supported by the project ``TriCS'' (ZI 1516/1-1) of the German Research Foundation (DFG) and the Saarbrücken Graduate School of Computer Science}}
\author{Daniel Neider\inst{1}, Alexander Weinert\inst{2} and Martin Zimmermann\inst{3}}

\institute{Max Planck Institute for Software Systems, 67663 Kaiserslautern, Germany \\ \email{neider@mpi-sws.org} \and
German Aerospace Center (DLR), Intelligent and Distributed Systems, Linder Höhe, 51147 Köln, Germany\\
  \email{alexander.weinert@dlr.de} \and
  University of Liverpool, Liverpool L69 3BX, United Kingdom\\
\email{martin.zimmermann@liverpool.ac.uk}
  }

\begin{document}

\maketitle

\begin{abstract}
Recently, Dallal, Neider, and Tabuada studied a generalization of the classical game-theoretic model used in program synthesis, which additionally accounts for unmodeled intermittent disturbances.
In this extended framework, one is interested in computing optimally resilient strategies, i.e., strategies that are resilient against as many disturbances as possible.
Dallal, Neider, and Tabuada showed how to compute such strategies for safety specifications.

In this work, we compute optimally resilient strategies for a much wider range of winning conditions and show that they do not require more memory than winning strategies in the classical model. Our algorithms only have a polynomial overhead in comparison to the ones computing winning strategies.
In particular, for parity conditions, optimally resilient strategies are positional and can be computed in quasipolynomial time.
\end{abstract}

\section{Introduction}
\label{sec_intro}
Reactive synthesis is an exciting and promising approach to solving a crucial problem, whose importance is ever-increasing due to ubiquitous deployment of embedded systems: obtaining correct and verified controllers for safety-critical systems. Instead of an engineer programming a controller by hand and then verifying it against a formal specification, synthesis automatically constructs a correct-by-construction controller from the given specification (or reports that no such controller exists). 

Typically, reactive synthesis is modeled as a two-player zero-sum game on a finite graph that is played between the system, which seeks to satisfy the specification, and its environment, which seeks to violate it.
Although this model is well understood, there are still multiple obstacles to overcome before synthesis can be realistically applied in practice. These obstacles include not only the high computational complexity of the problem, but also more fundamental ones.
Among the most prohibitive  issues in this regard is the need for a complete model of the interaction between the system and its environment, including an accurate model of the environment, the actions available to both players, as well as the effects of these actions. 

This modeling task often places an insurmountable burden on engineers as the environments in which real-life controllers are intended to operate tend to be highly complex or not fully known at design time. Also, when a controller is deployed in the real world, a common source of errors is a mismatch between the controller's intended result of an action and the actual result.
Such situations arise, e.g., in the presence of disturbances, when the effect of an action is not precisely known, or when the intended control action of the controller cannot be executed, e.g., when an actuator malfunctions. By a slight abuse of notation from control theory, such errors are subsumed under the generic term \emph{disturbance} (cf.~\cite{DBLP:conf/cdc/DallalNT16}).

To obtain controllers that can handle disturbances, one has to yield control over their occurrence to the environment. However, due to the antagonistic setting of the two-player zero-sum game, this would allow the environment to violate the specification by causing disturbances at will. Overcoming this requires the engineer to develop a realistic disturbance model, which is a highly complex task, as such disturbances are assumed to be rare events. Also, incorporating such a model into the game leads to a severe blowup in the size of the game, which can lead to intractability due to the high computational complexity of synthesis. 

To overcome these fundamental difficulties, Dallal, Neider, and Tabuada~\cite{DBLP:conf/cdc/DallalNT16} proposed a conceptually simple, yet powerful extension of infinite games termed ``games with unmodeled intermittent disturbances''.
Such games are played similarly to classical infinite games: two players, called Player~$0$ and Player~$1$, move a token through a finite graph, whose vertices are partitioned into vertices under the control of Player~$0$ and Player~$1$, respectively; the winner is declared based on a condition on the resulting play.
In contrast to classical games, however, the graph is augmented with additional \emph{disturbance edges} that originate in vertices of Player~$0$ and may lead to any other vertex. Moreover, the mechanics of how Player~$0$ moves is modified: whenever she moves the token, her move might be overridden, and the token instead moves along a disturbance edge.
This change in outcome implicitly models the occurrence of a disturbance---the intended result of the controller and the actual result differ---but it is not considered to be antagonistic.
Instead, the occurrence of a disturbance is treated as a rare event without any assumptions on frequency, distribution, etc.
This approach very naturally models the kind of disturbances typically occurring in control engineering~\cite{DBLP:conf/cdc/DallalNT16}.

As a non-technical example, consider a scenario with three siblings, Alice, Bob, and Charlie, and their father, Donald.
He repeatedly asks Alice to fetch water from a well using a jug made of clay.
Alice has three ways to fulfill that task:
she may get the water herself or she may delegate it to either Bob or Charlie.
In a simple model, the outcome of these strategies is identical: Donald's request for water is fulfilled.
This is, however, unrealistic, as this model ignores the various ways that the execution of the strategies may go wrong.
By modeling the situation as a game with disturbances, we obtain a more realistic model.

If Alice gets the jug herself, no disturbance can occur: she controls the outcome completely.
If she delegates the task to Bob, the older of her brothers, Donald may get angry with her for not fulfilling her duties herself, which should not happen infinitely often.
Finally, if she delegates the task to her younger brother Charlie, he might drop and break the jug, which would be disastrous for Alice.

These strategies can withstand different numbers of disturbances: the first strategy does not offer any possibility for disturbances, while infinitely many (a single) disturbance cause Alice to lose when using the second (the third) strategy.
This model captures the intuition about Donald's and Charlie’s behavior: both events occur non-antagonistically and their frequency is unknown.

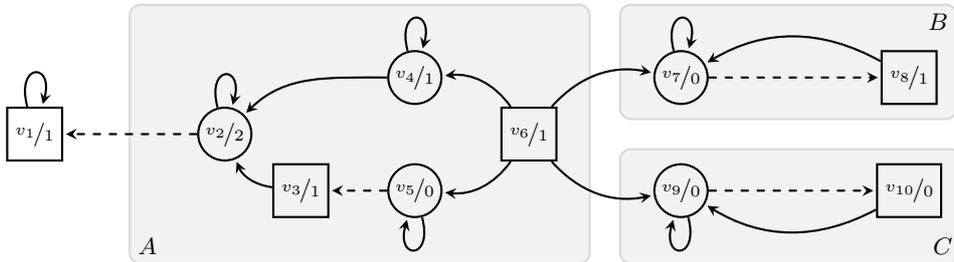
\begin{figure*}
	\centering
	\scalebox{.93}{
	\begin{tikzpicture}
		\ParityVertexOne{6}{v_6}{1}{(-1,0)}
		\ParityVertexZero{5}{v_4}{1}{(-2.5,.75)}
		\ParityVertexZero{4}{v_5}{0}{(-2.5,-.75)}
		\ParityVertexOne{3}{v_3}{1}{(-4,-.75)}
		\ParityVertexZero{2}{v_2}{2}{(-5,0)}
		\ParityVertexOne{1}{v_1}{1}{(-7.5,0)}
		\ParityVertexZero{8}{v_7}{0}{(1,.75)}
		\ParityVertexOne{9}{v_8}{1}{(4,.75)}
		\ParityVertexZero{10}{v_{9}}{0}{(1,-.75)}
		\ParityVertexOne{12}{v_{10}}{0}{(4,-.75)}
	
		\path
			(6) edge[bend right] (5) edge[bend left] (4) edge[bend left] (8) edge[bend right] (10)
			(5) edge [loop above] (5) edge[out=180,in=45] (2)
			(4) edge [loop below] (4) edge [fault] (3)
			(3) edge[bend left] (2)
			(2) edge [loop above] (2) edge [fault] (1)
			(1) edge [loop above] (1)
			(8) edge [fault] (9) edge [loop above] (8)
			(10) edge [fault] (12) edge [loop below] (10)
			(9) edge [bend right] (8)
			(12) edge [bend left] (10);
			
		\begin{pgfonlayer}{background}
		
			\coordinate (center-left) at ($(8) ! 0.5 ! (10)$);
			\coordinate (center) at ($(center-left) ! .5 ! (6)$);
			\coordinate (northwest) at ($(9.north east) + (.25cm,.6cm)$);
			\coordinate (southwest) at ($(12.south east) + (.25cm,-.6cm)$);
			\coordinate (east) at ($(2) ! .5 ! (1)$);
			\coordinate (southeast-low) at (3.south -| east);
			\coordinate (southeast) at ($(southeast-low) - (0,.6cm)$);
			\coordinate (northeast-low) at (5.north -| east);
			\coordinate (northeast) at ($(northeast-low) + (0,.6cm)$);
			
			\draw[black!15,thick,rounded corners,fill=black!5]
				(northeast) -| ($(center) - (.2cm,0)$) |- (southeast) -- cycle;
			\node[anchor=south west] at (southeast) {$A$};
			\draw[black!15,thick,rounded corners,fill=black!5]
				($(center) + (.2cm,.2cm)$) -| (northwest) -| cycle;
			\node[anchor=north east] at (northwest) {$B$};
			\draw[black!15,thick,rounded corners,fill=black!5]
				(southwest) -| ($(center) + (.2cm,-.2cm)$) -| cycle;
			\node[anchor=south east] at (southwest) {$C$};
			
		\end{pgfonlayer}

	\end{tikzpicture}
	}
	\caption{A (max-) parity game with disturbances. Disturbance edges are drawn as dashed arrows. Vertices are labeled with both a name and a color. Vertices under control of Player~$0$ are drawn as circles, while vertices under control of Player~$1$ are drawn as rectangles.}
	\label{fig:example}
\end{figure*}

This non-antagonistic nature of disturbances is different from existing approaches in the literature and causes many interesting phenomena that do not occur in the classical theory of infinite graph-based games.
In Figure~\ref{fig:example}, we show an example of a parity game with disturbances that already exhibits some of these phenomena.
In that parity game, vertices are labeled with non-negative integers, so-called colors, and Player~$0$ wins if the highest color seen infinitely often is even.
For the sake of readability and conciseness, the parity game in Figure~\ref{fig:example} does not model the example given in natural language above, but is rather constructed to showcase properties of games with disturbances.

Consider, for instance, vertex $v_2$. In the classical setting without disturbances, Player~$0$ wins every play reaching $v_2$ by simply looping in this vertex forever (since the highest color seen infinitely often is even). However, this is no longer true in the presence of disturbances: a disturbance in $v_2$ causes a play to proceed to vertex $v_1$, from which Player~$0$ can no longer win. In vertex $v_7$, Player~$0$ is in a similar, yet less severe situation: she wins every play with finitely many disturbances but loses if infinitely many disturbances occur. 
Finally, vertex $v_9$ falls into a third category: from this vertex, Player~$0$ wins every play even if infinitely many disturbances occur.
In fact, disturbances partition the set of vertices from which Player~$0$ can guarantee to win into three disjoint regions (indicated as shaded boxes in Figure~\ref{fig:example}): (A) vertices from which she can win if at most a fixed finite number of disturbances occur, (B) vertices from which she can win if any finite number of disturbances occurs but not if infinitely many occur, and (C) vertices from which she can win even if infinitely many disturbances occur.

The observation above gives rise to a question that is both theoretically interesting and practically important: if Player~$0$ can tolerate different numbers of disturbances from different vertices, how should she play to be \emph{resilient}\footnote{We have deliberately chosen the term \emph{resilience} so as to avoid confusion with the already highly ambiguous notions of \emph{robustness} and \emph{fault tolerance}.} to as many disturbances as possible, i.e., to tolerate as many disturbances as possible but still win?
Put slightly differently, disturbances induce an order on the space of winning strategies (``a winning strategy is better if it is more resilient''), and the natural problem is to compute optimally resilient winning strategies, yielding optimally resilient controllers.
Note that this is in contrast to the classical theory of infinite games, where the space of winning strategies is unstructured.

Dallal, Neider, and Tabuada~\cite{DBLP:conf/cdc/DallalNT16} have solved the problem of computing optimally resilient winning strategies for safety games. Their approach exploits the existence of maximally permissive winning strategies in safety games~\cite{BernetJaninWalukiewicz02}, which allows Player~$0$ to avoid ``harmful'' disturbance edges during a play.
In games with more expressive winning conditions, however, this is no longer possible, as witnessed by vertex $v_4$ in the example of Figure~\ref{fig:example}: although Player~$0$ can avoid a disturbance edge by looping in $v_4$ forever, she needs to move to $v_2$ eventually in order to see an even color (otherwise she loses), thereby risking to lose if a disturbance occurs.
In fact, the problem of constructing optimally resilient winning strategies for games other than safety games has been left open by Dallal, Neider, and Tabuada. In this work, we solve this problem for a large class of infinite games, including parity games.

\paragraph{Our Contributions.}
\label{para_ourcontribs}
In Section~\ref{sec_prelims}, we introduce the concept of \emph{resilience}, which captures for each vertex how many disturbances need to occur for Player~$0$ to lose. This generalizes the notion of determinacy and allows us to derive optimally resilient winning strategies. 

Our main result is an algorithm for computing the resilience of vertices and optimally resilient winning strategies, which we present in Section~\ref{sec_results}. This algorithm requires the game to have a prefix-independent winning condition, to be determined, and all its subgames to be (classically) solvable. The latter two conditions are necessary, as resilience generalizes determinacy and computing optimally resilient strategies generalizes solving games. We discuss these assumptions in Section~\ref{sec_discussion}.

 The algorithm uses solvers for the underlying game without disturbances as a subroutine, which it invokes a linear number of times on various subgames. For many winning conditions, the time complexity of our algorithm thus falls into the same complexity class as solving the original game without disturbances, e.g., we obtain a quasipolynomial algorithm for parity games with disturbances, which matches the currently best known upper bound for classical parity games.
Stated differently, if the three assumptions above are satisfied by a winning condition, then computing the resilience and optimally resilient strategies is not harder than determining winning regions and winning strategies (ignoring a polynomial overhead).

Our algorithm requires the winning condition of the game to be prefix-independent. We also show how to overcome this restriction by generalizing the classical notion of game reductions to the setting of games with disturbances. As a consequence, via reductions, our algorithm can be applied to prefix-dependent winning conditions. We discuss details in Section~\ref{sec_discussion}.

Altogether, we have generalized the original result of Dallal, Neider, and Tabuada from safety games to all games which are algorithmically solvable, in particular all $\omega$-regular games.

Finally, in Section~\ref{sec_outlook}, we discuss further phenomena that arise in the presence of disturbances. Amongst others, we illustrate how the additional goal of avoiding disturbances whenever possible affects the memory requirements of strategies. Similarly, we exhibit a tradeoff between resilience and the (semantic) quality of strategies in quantitative games. Moreover, we raise the question of how benevolent disturbances can be leveraged to recover from losing a play. However, an in-depth investigation of these phenomena is outside the scope of this paper and left for future work.

\section{Preliminaries}
\label{sec_prelims}
For notational convenience, we employ some  ordinal notation à la von Neumann: the non-negative integers are defined inductively as $0 = \emptyset$  and $n+1 = n \cup \set{n}$. Now, the first limit ordinal is $\omega = \set{0,1,2, \ldots}$, the set of the non-negative integers. The next two successor ordinals are $\omega+1 = \omega \cup \set{\omega}$ and $\omega+2 = \omega+1 \cup \set{\omega+1}$. These ordinals are ordered by set inclusion, i.e., we have $0 < 1< 2 < \cdots < \omega < \omega+1 < \omega+2$. For convenience of notation, we also denote the cardinality of $\omega$ by $\omega$.

\subsection{Infinite Games with Disturbances.}
\label{subsec_games}
An \emph{arena (with unmodeled disturbances)}~$\arena = (V, V_0, V_1, E, D)$ consists of a finite directed graph~$(V, E)$, a partition~$\set{V_0, V_1}$ of $V$ into the set of vertices~$V_0$ of Player~$0$ (denoted by circles) and the set of vertices of Player~$1$ (denoted by squares), and a set~$D \subseteq V_0 \times V$ of \emph{disturbance edges} (denoted by dashed arrows). Note that only vertices of Player~$0$ have outgoing disturbance edges. We require that every vertex~$v \in V$ has a successor $v'$ with $(v,v') \in E$ to avoid finite plays.

A \emph{play} in $\arena$ is an infinite sequence
$ \rho = (v_0, b_0) (v_1, b_1) (v_2, b_2) \cdots \in (V\times\set{0,1})^\omega $
such that $b_0 = 0$ and for all $j>0$: $b_j = 0$ implies $(v_{j-1}, v_j) \in E$, and $b_j = 1$ implies $(v_{j-1}, v_j) \in D$. Hence, the additional bits~$b_j$ for~$j > 0$ denote whether a standard or a disturbance edge has been taken to move from $v_{j-1}$ to $v_j$, while $b_0$ is always zero.
We say $\rho$ starts in $v_0$. A play prefix~$(v_0, b_0) \cdots (v_j, b_j)$ is defined similarly and ends in $v_j$. The number of disturbances in a play~$\rho = (v_0, b_0) (v_1, b_1) (v_2, b_2) \cdots$ is $\disturbances(\rho) = \size{\set{j \in \omega \mid b_j = 1}}$, which is either some $k \in \omega$ (if there are finitely many disturbances, namely $k$) or it is equal to $\omega$ (if there are infinitely many). A play~$\rho$ is \emph{disturbance-free}, if $\disturbances(\rho) = 0$.

 A \emph{game (with unmodeled disturbances)}, denoted by~$\game = (\arena, \wincond)$, consists of an arena~$\arena = (V, V_0, V_1, E, D)$ and a winning condition~$\wincond \subseteq V^\omega$. A play~$\rho = (v_0, b_0) (v_1, b_1) (v_2, b_2) \cdots$ is winning for Player~$0$, if  $v_0 v_1 v_2 \cdots \in \wincond$, otherwise it is winning for Player~$1$. Hence, winning is oblivious to occurrences of disturbances. A winning condition~$\wincond$ is \emph{prefix-independent} if for all $\rho \in V^\omega$ and all $w \in V^*$ we have $\rho \in \wincond$ if and only if $w\rho \in \wincond$. If $\wincond$ is not prefix-independent, then it is called prefix-dependent.

 In examples, we often use the parity condition, the canonical $\omega$-regular winning condition. Let $\col \colon V \rightarrow \omega$ be a coloring of a set~$V$ of vertices. The (max-) parity condition
\[\parity(\col) = \set{ v_0 v_1 v_2 \cdots \in V^\omega \mid \limsup \col(v_0) \col(v_1) \col(v_2) \cdots \text{ is even}}\]
requires the maximal color occurring infinitely often during a play to be even. A game~$(\arena, \wincond)$ is a parity game, if $\wincond = \parity(\col)$ for some coloring~$\col$ of the vertices of $\arena$. In figures, we label a vertex~$v$ with color~$c$ by $\nicefrac{v}{c}$.

In our proofs we make use of the safety condition
\[
\safety(U) = \set{v_0 v_1 v_2 \cdots \in V^\omega \mid v_j \notin U \text{ for every }j \in \omega}
\]
for a given set~$U \subseteq V$ of {unsafe} vertices. It requires Player~$0$ to only visit safe vertices, i.e., Player~$1$ wins a play if it visits at least one unsafe vertex. Note that due to notational convenience, we specify a safety condition by giving the \emph{unsafe} vertices instead of the \emph{safe} ones, i.e., $V \setminus U$, which is more common. 

A strategy for Player~$i \in \set{0,1}$ is a function~$\sigma \colon V^*V_i \rightarrow V$ such that $(v_j, \sigma(v_0 \cdots v_j)) \in E$ holds for every $v_0 \cdots v_j \in V^*V_i$. A play~$(v_0,b_0) (v_1, b_1) (v_2,b_2) \cdots $ is consistent with $\sigma$, if $v_{j+1} = \sigma(v_0 \cdots v_j)$ for every $j$ with $v_j \in V_i$ and $b_{j+1} = 0$, i.e., if the next vertex is the one prescribed by the strategy unless a disturbance edge is used. 

\begin{remark}
A  strategy~$\sigma$ does not have access to the bits indicating whether a disturbance occurred or not. However, this is not a restriction for Player~$0$: let $(v_0,b_0) (v_1, b_1) (v_2,b_2) \cdots $ be a play with $b_j = 1$ for some $j > 0$. We say that this disturbance is consequential (w.r.t.\ $\sigma$), if $v_j \neq \sigma(v_0 \cdots v_{j-1})$, i.e., if the disturbance transition~$(v_{j-1}, v_j)$ traversed by the play did not lead to the vertex the strategy prescribed. Such consequential disturbances can be detected by comparing the actual vertex~$v_j$ to $\sigma$'s output~$\sigma(v_0 \cdots v_{j-1})$. Hence, the bits~$b_j$ denoting consequential disturbances (w.r.t.~$\sigma$) can be reconstructed by observing the sequence of vertices and by having access to the strategy~$\sigma$.

On the other hand, inconsequential disturbances can just be ignored. In particular, the number of consequential disturbances is always at most the number of disturbances during each play.
\end{remark}

\subsection{Positional and Finite-state Strategies.}
\label{subsec_finitestatestrategies}
Fix a game~$(\arena, \wincond)$ with $\arena = (V, V_0, V_1, E, D)$. A strategy~$\sigma$ for Player~$i$ is positional, if $\sigma(v_0 \cdots v_j) = \sigma(v_j)$ for all $v_0 \cdots v_j \in V^*V_i$, i.e., the output of $\sigma$ only depends on the last vertex. 

 A memory structure for $\arena$ is a triple~$\mem = (M, \init, \update)$ where $M$ is a finite set of memory states, $\init \colon V \rightarrow M$ is the initialization function, and $\update \colon M \times V \rightarrow M$ is the memory update function. 

The update function can be extended to finite play prefixes: $\update^+(v) = \init(v)$ and $\update^+(wv) = \update(\update^+(w), v)$ for $w \in V^+$ and $v \in V$.
A next-move function $\nxt \colon V_i \times M \rightarrow V$ for Player~$i$ has to satisfy $(v, \nxt(v, m)) \in E$ for all $v \in V_i$ and all $m \in M$.
It induces a strategy~$\sigma$ for Player~$i$ with memory~$\mem$ via $\sigma(v_0\cdots v_j) = \nxt(v_j, \update^+(v_0 \cdots v_j))$.

We say that a strategy~$\sigma$ is implementable by a memory structure~$\mem$, if there is a next-move function~$\nxt$ such that $\mem$ and $\nxt$ induce $\sigma$. If $\sigma$ is implementable by some memory structure, then we call $\sigma$ finite-state.

\subsection{Infinite Games without Disturbances.}
\label{subsec_gameswithoutdisturbances}
We can characterize the classical notion of infinite games, i.e., those without disturbances, (see, e.g.,~\cite{GraedelThomasWilke02}) as a  special case of games with disturbances. Let $\game$ be a game with vertex set~$V$. A strategy~$\sigma$ for Player~$i$ in $\game$ is a winning strategy for her from $v \in V$, if every disturbance-free play that starts in $v$ and that is consistent with $\sigma$ is winning for Player~$i$.

The winning region~$\winreg_i(\game)$ of Player~$i$ in $\game$ contains those vertices~$v \in V$ from which Player~$i$ has a winning strategy. Thus, the winning regions of $\game$ are independent of the disturbance edges, i.e., we obtain the classical notion of infinite games. We say that Player~$i$ wins $\game$ from $v$, if $v \in \winreg_i(\game)$. Solving a game amounts to determining its winning regions. Note that every game has disjoint winning regions. In contrast, a game is determined, if every vertex is in either winning region.

\subsection{Resilient Strategies.}
\label{subsec_gameswithdisturbances}
Let $\game$ be a game with vertex set~$V$ and let $\alpha \in \omega+2$. A strategy~$\sigma$ for Player~$0$ in $\game$ is \emph{$\alpha$-resilient from~$v \in V$} if every play~$\rho$ that starts in $v$, that is consistent with $\sigma$, and with~$\disturbances(\rho) < \alpha$, is winning for Player~$0$. Thus, a $k$-resilient strategy with $k \in \omega$ is winning even under at most $k-1$ disturbances, an $\omega$-resilient strategy is winning even under any finite number of disturbances, and an $(\omega+1)$-resilient strategy is winning even under infinitely many disturbances. 

\begin{remark}
\label{remark_resilienceprops}
Let $v$ be a vertex.
\begin{enumerate}
	\item\label{remark_resilienceprops_mono} Let $\alpha , \alpha' \in \omega+2$ with $\alpha > \alpha'$. If a strategy is $\alpha$-resilient from $v$, then it is also $\alpha'$-resilient from $v$.
	\item\label{remark_resilienceprops_zeroresil} Every strategy is $0$-resilient from $v$.
	\item\label{remark_resilienceprops_winning} A strategy is~$1$-resilient from~$v$ if and only if it is winning for Player~$0$ from~$v$. 
\end{enumerate}
\end{remark}
We define the resilience of a vertex~$v$ of $\game$ as 
\[
r_\game(v) = \sup \set{ \alpha\in\omega+2 \mid \text{Player~$0$ has an $\alpha$-resilient strategy for $\game$ from $v$} }.
\]
Note that the definition is not antagonistic, i.e., it is not defined via strategies of Player~$1$. Nevertheless, due to the remarks above, resilient strategies generalize winning strategies. 

\begin{lemma}
\label{lemma_winningregionsvsresilience}
Let $\game$ be a game and $v$ a vertex of $\game$.
\begin{enumerate}
\item\label{lemma_winningregionsvsresilience_zero} $r_\game(v) > 0$ if and only if $v \in \winreg_0(\game)$.
\item\label{lemma_winningregionsvsresilience_one} If $\game$ is determined, then $r_\game(v) = 0$ if and only if $v \in \winreg_1(\game)$.
\end{enumerate}
\end{lemma}

\begin{proof}
\ref{lemma_winningregionsvsresilience_zero}.) The resilience of $v$ is greater than zero if and only if Player~$0$ has a $1$-resilient strategy from $v$ due to Item~\ref{remark_resilienceprops_mono} of Remark~\ref{remark_resilienceprops}. The latter condition is equivalent to Player~$0$ having a winning strategy for $\game$ from $v$, i.e., to $v \in \winreg_0(\game)$, due to Item~\ref{remark_resilienceprops_winning} of Remark~\ref{remark_resilienceprops}.

\ref{lemma_winningregionsvsresilience_one}.) Due to Items~\ref{remark_resilienceprops_mono} and \ref{remark_resilienceprops_winning} of Remark~\ref{remark_resilienceprops}, the resilience of $v$ is zero if and only if Player~$0$ has no winning strategy for $\game$ from $v$, i.e., $v \notin \winreg_0(\game)$. Due to determinacy, this is equivalent to $v \in \winreg_1(\game)$.	\qed
\end{proof}
Note that determinacy is a necessary condition for Item~\ref{lemma_winningregionsvsresilience_one}. In an undetermined game, the vertices that are in neither winning region have resilience zero, due to Item~\ref{lemma_winningregionsvsresilience_zero}, but are in particular not in $\winreg_1(\game)$.

A strategy~$\sigma$ is \emph{optimally resilient}, if it is $r_\game(v)$-resilient from every vertex~$v$. Every such strategy is a \emph{uniform} winning strategy for Player~$0$, i.e., a strategy that is winning from every vertex in her winning region. Hence, positional optimally resilient strategies can only exist in games which have uniform positional winning strategies for Player~$0$.

Our goal is to determine the mapping~$r_\game$ and to compute an optimally resilient strategy.

\section{Computing Optimally Resilient Strategies}
\label{sec_results}
To compute optimally resilient strategies, we first characterize the vertices of finite resilience in Subsection~\ref{subsec_finiteresil}. All other vertices either have resilience~$\omega$ or $\omega+1$. To distinguish between these possibilities, we show how to determine the vertices with resilience~$\omega+1$ in Subsection~\ref{subsec_resilomegaplus1}. In Subsection~\ref{subsec_strategies}, we show how to compute optimally resilient strategies using the results of the first two subsections. 
We only consider prefix-independent winning conditions in Subsections~\ref{subsec_finiteresil} and \ref{subsec_strategies}. In Section~\ref{sec_discussion}, we show how to overcome this restriction.

\subsection{Characterizing Vertices of Finite Resilience}
\label{subsec_finiteresil}
Our goal in this subsection is to characterize vertices with finite resilience in a game with prefix-independent winning condition, i.e., those vertices from which Player~$0$ can win even under~$k-1$ disturbances, but not under $k$ disturbances, for some $k \in \omega$. 

To illustrate our approach, consider the parity game in Figure~\ref{fig:example} (on Page~\pageref{fig:example}), which is determined and has a prefix-independent winning condition. The winning region of Player~$1$ only contains the vertex~$v_1$. Thus, by Lemma~\ref{lemma_winningregionsvsresilience}, $v_1$ is the only vertex with resilience zero, every other vertex has a larger resilience. 

Now, consider the vertex~$v_2$, which has a disturbance edge leading into the winning region of Player~$1$. Due to this edge, $v_2$ has resilience at most one. This implies, as argued above, that~$v_2$ has resilience precisely one. The unique disturbance-free play starting in $v_1$ is consistent with every strategy for Player~$0$ and violates the winning condition. Due to prefix-independence, prepending the disturbance edge does not change the winner and consistency with every strategy for Player~$0$. Hence, this play witnesses that $v_2$ has resilience at most one, while $v_2$ being in Player~$0$'s winning region yields the matching lower bound.
However, $v_2$ is the only vertex to which this reasoning applies.
Now, consider $v_3$: from here, Player~$1$ can force a play to visit $v_2$ using a standard edge. Thus, $v_3$ has resilience one as well. Again, this is the only vertex to which this reasoning is applicable. 

In particular, from $v_4$, Player~$0$ can avoid reaching the vertices for which we have already determined the resilience by using the self loop. However, this comes at a steep price for her: doing so results in a losing play, as the color of $v_4$ is odd. Thus, if she wants to have a chance at winning, she has to take a risk by moving to $v_2$, from which she has a $1$-resilient strategy, i.e., one that is winning if no more disturbances occur. For this reason, $v_4$ has resilience one as well. The same reasoning applies to $v_6$: Player~$1$ can force the play to $v_4$ and from there Player~$0$ has to take a risk by moving to $v_2$.

The vertices~$v_3$, $v_4$, and $v_6$ share the property that Player~$1$ can either enforce a play violating the winning condition or reach a vertex with already determined finite resilience. These three vertices are the only ones currently satisfying this property. They all have resilience one since Player~$1$ can enforce to reach a vertex of resilience one, but he cannot enforce reaching a vertex of resilience zero. Now, we can also determine the resilience of $v_5$: the disturbance edge from $v_5$ to $v_3$ witnesses it being two.

Afterwards, these two arguments no longer apply to new vertices: no disturbance edge leads from a vertex~$v \in \set{v_7, \ldots, v_{10}} $ to some vertex whose resilience is already determined and Player~$0$ has a winning strategy from each such $v$ that additionally avoids vertices whose resilience is already determined. Thus, our reasoning cannot determine their resilience. This is consistent with our goal, as all four vertices have non-finite resilience: $v_7$ and $v_8$ have resilience~$\omega$ and $v_9$ and $v_{10}$ have resilience~$\omega+1$.
Our reasoning here cannot distinguish these two values.
We solve this problem in Subsection~\ref{subsec_resilomegaplus1}.

We now formalize the reasoning sketched above: starting from the vertices in Player~$1$'s winning region having resilience zero, we use a so-called disturbance update and a risk update to determine all vertices of finite resilience.
A disturbance update computes the resilience of vertices having a disturbance edge to a vertex whose resilience is already known (such as vertices $v_2$ and $v_5$ in the example of Figure~\ref{fig:example}). A risk update, on the other hand, determines the resilience of vertices from which either Player~$1$ can force a visit to a vertex with known resilience (such as vertices $v_3$ and $v_6$) or Player~$0$ needs to move to such a vertex in order to avoid losing (e.g., vertex $v_4$).
To simplify our proofs, we describe both as monotone operators updating partial rankings mapping vertices to $\omega$, which might update already defined values. We show that applying these updates in alternation eventually yields a stable ranking that indeed characterizes the vertices of finite resilience. 

Throughout this section, we fix a game~$\game = (\arena, \wincond)$ with $\arena = (V, V_0, V_1, E, D)$ and prefix-independent~$\wincond \subseteq V^\omega$ satisfying the following condition: the game $(\arena, \wincond \cap \safety(U))$ is determined for every $U \subseteq V$. We discuss this requirement in Section~\ref{sec_discussion}.

A ranking for $\game$ is a partial mapping~$r \colon V \dashrightarrow \omega$. The domain of $r$ is denoted by $\dom(r)$, its image by $\im(r)$. Let $r$ and $r'$ be two rankings. We say that $r'$ refines $r$ if $\dom(r') \supseteq \dom(r)$ and if $r'(v) \le r(v)$ for all $v \in \dom(r)$. A ranking~$r$ is sound, if we have $r(v) = 0$ if and only if $v \in \winreg_1(\game)$ (cf.~Lemma~\ref{lemma_winningregionsvsresilience}). 

Let $r$ be a ranking for $\game$. We define the ranking~$r'$ as 
\[
r'(v) = \min \bigl( \set{r(v)} \cup \set{r(v')+1 \mid v' \in \dom(r) \text{ and } (v,v') \in D} \bigr),
\]
where $\set{r(v)} = \emptyset$ if $v \notin \dom(r)$, and $\min \emptyset$ is undefined (causing $r'(v)$ to be undefined). We call $r'$ the \emph{disturbance update} of $r$.

\begin{lemma}
\label{lemma_disturbanceupdateproperties}
The disturbance update~$r'$ of a sound ranking~$r$ is sound and refines~$r$. 
\end{lemma}

\begin{proof}
As the minimization defining $r'(v)$ ranges over a superset of $\set{r(v)}$, we have $r'(v) \le r(v)$ for every $v \in \dom(r)$. This immediately implies refinement. From this inequality, we also obtain $r'(v) = 0$ for every $v \in \winreg_1(\game)$, due to soundness of $r$. Finally, consider some $v \in \winreg_0(\game)$. Then, $r(v) >0$ by soundness of $r$. Thus, $r'(v) > 0$ as well, as both $r(v)$ and each $r(v') + 1$ are greater than zero. Altogether, $r'$ is sound as well. \qed 
\end{proof}
Again, let $r$ be a ranking for $\game$. For every $k \in \im(r)$ let
\[
A_k = \winreg_1(\arena, \wincond \cap \safety(\set{v \in \dom(r) \mid r(v) \le k}))
\]
be the winning region of Player~$1$ in the game where he either wins by reaching a vertex~$v$ with $r(v) \le k$ or by violating the winning condition of $\game$. 
Now, define $r'(v) = \min \set{k \mid v \in A_k}$, where $\min \emptyset$ is again undefined. We call $r'$ the \emph{risk update} of $r$.

\begin{lemma}
\label{lemma_riskupdateproperties}
The risk update~$r'$ of a sound ranking~$r$ is sound and refines $r$. 
\end{lemma}

\begin{proof}
We show $r'(v) \le r(v)$ for every $v \in \dom(r)$, which implies both refinement and $r'(v) = 0$ for every $v \in \winreg_1(\game)$, as argued in the proof of Lemma~\ref{lemma_disturbanceupdateproperties}.

Thus, let $v \in \dom(r)$. Trivially, $v \in \set{v' \in \dom(r) \mid r(v') \le r(v)}$. Thus, Player~$1$ wins the game
$(\arena, \wincond \cap \safety(\set{v' \in \dom(r) \mid r(v') \le r(v) }))$ from $v$ by violating the safety condition right away. Hence, $v \in A_{r(v)}$ and thus $r'(v) \le r(v)$.

To complete the proof of soundness of $r'$, we just have to show $r'(v) > 0$ for every $v \in \winreg_0(\game)$. Towards a contradiction, assume $r'(v) = 0$, i.e., $v \in A_0$. Thus, Player~$1$ has a strategy~$\tau$ from $v$ that ensures that either the winning condition is violated or that a vertex~$v'$ with $r(v') = 0$ is reached, i.e., $v' \in \winreg_1(\game)$ by soundness of $r$. Hence, Player~$1$ has a winning strategy~$\tau_{v'}$ for $\game$ from every such $v'$. This implies that he also has a winning strategy from~$v$:  play according to $\tau$ until a vertex~$v'$ with $r(v')=0$ is reached. From there, mimic $\tau_{v'}$ when starting from $v'$. Every resulting disturbance-free play has a suffix that violates the winning condition~$\wincond$. Thus, by prefix-independence, the whole play violates $\wincond$ as well, i.e., it is winning for Player~$1$. Thus, $v \in \winreg_1(\game)$, which yields the desired contradiction, as winning regions are always disjoint.\qed 
\end{proof}

Let $r_0$ be the unique sound ranking with domain~$\winreg_1(\game)$, i.e., $r_0$ maps exactly the vertices in Player~$1$'s winning region to zero, all others are undefined. Starting with $r_0$, we inductively define a sequence of rankings~$(r_j)_{j \in \omega}$ such that $r_{j}$ for an odd (even) $j >0$ is the disturbance (risk) update of $r_{j-1}$, i.e., we alternate between disturbance and risk updates.

Due to refinement, the $r_j$ eventually stabilize, i.e., there is some $j_0$ such that $r_j = r_{j_0}$ for all $j \ge j_0$. Define $r^* = r_{j_0}$. Due to $r_0$ being sound and by Lemma~\ref{lemma_disturbanceupdateproperties} and Lemma~\ref{lemma_riskupdateproperties}, each $r_j$, and $r^*$ in particular, is sound. If $v \in \dom(r^*)$, let $j_v$ be the minimal $j$ with $v \in \dom(r_j)$; otherwise, $j_v$ is undefined.

\begin{lemma}
\label{lemma_ranktermination}
If $v \in \dom(r^*)$, then $r_{j_v}(v) = r_j(v)$ for all $j \ge j_v$.
\end{lemma}

\begin{proof}
We show the following stronger result for every $v \in \dom(r^*)$:
\begin{itemize}
	\item If~$j_v$ is odd, then $r_j(v) = \frac{j_v+1}{2}$ for every $j \ge j_v$.
	\item If~$j_v$ is even, then $r_j(v) = \frac{j_v}{2}$ for every $j \ge j_v$.
\end{itemize}

The disturbance update increases the maximal rank by at most one and the risk update does not increase the maximal rank at all. Furthermore, due to refinement, the rank of $v$ is set and then it cannot increase. Hence, we obtain $r_j(v) \leq \frac{j_v+1}{2}$ and $r_j(v) \leq \frac{j_v}{2}$ for odd and even~$j_v$, respectively. In the remainder of the proof, we show a matching lower bound.

We say that a vertex~$v$ is updated to~$k \in \omega$ in~$r_j$ if~$r_j(v) = k$ and either~$v \notin \dom(r_{j-1})$ or both~$v \in \dom(r_{j-1})$ and~$r_{j-1}(v) \neq k$ (here, $r_{-1}$ is the unique ranking with empty domain). Note that as part of the proof, we have to show that the second case never occurs.

Now, we show the following by induction over~$j$, which implies the matching lower bound.
\begin{itemize}
	\item If~$j$ is odd, then no~$v$ is updated in~$r_j$ to some~$k < \frac{j+1}{2}$.
	\item If~$j$ is even, then no~$v$ is updated in~$r_j$ to some~$k < \frac{j}{2}$.
\end{itemize}
	
	For~$j = 0$, we have~$\frac{j}{2} = 0$, and
	clearly, no vertex is assigned a negative rank by~$r_0$.
	For~$j = 1$ and~$j' = 2$, we obtain~$\frac{j+1}{2} = \frac{j'}{2} = 1$. As $r_0$, $r_1$, and $r_2$ are sound, neither~$r_1$ nor~$r_2$ update some~$v$ to zero.
	
	Now, let~$j > 2$ and first consider the case where~$j$ is odd.
	Towards a contradiction, assume that~$v \in V$ is updated in~$r_j$ to some value less than~$\frac{j+1}{2}$.
	Since~$j$ is odd,~$r_j$ is the disturbance update of~$r_{j-1}$. Further, as $v$ is updated in $r_j$, there exists some disturbance edge $(v, v') \in D$ such that $r_j(v) = r_{j-1}(v') + 1$.
	Thus, $r_{j-1}(v') < r_j(v) < \frac{j+1}{2}$, i.e., $r_{j-1}(v') \leq \frac{j+1}{2} - 2 = \frac{j-3}{2}$.
	First, we show~$r_{j-3}(v') = r_{j-2}(v') = r_{j-1}(v')$, i.e., the rank of $v'$ is stable during the last two updates.
	
	First assume towards a contradiction $r_{j-2}(v') \neq r_{j-1}(v')$.
	Then,~$v'$ is updated in~$r_{j-1}$ to some rank of at most~$\frac{j-3}{2}$, which is in turn smaller than~$\frac{j-1}{2}$, violating the induction hypothesis for~$j-1$.
	Hence, $r_{j-2}(v') = r_{j-1}(v')$.
	The same reasoning yields a contradiction to the assumption  $r_{j-3}(v') \neq r_{j-2}(v')$.
	Thus, we indeed obtain $r_{j-3}(v') = r_{j-2}(v') = r_{j-1}(v')$.
	
	Since~$r_{j-2}$ is the disturbance update of~$r_{j-3}$, we obtain~$r_{j-2}(v) \leq r_{j-3}(v') + 1 = r_{j-1}(v') + 1 = r_j(v)$.
	Due to refinement, we obtain $r_{j-2}(v) \geq r_{j}(v)$, i.e., altogether $r_{j-2}(v) = r_{j-1}(v) = r_{j}(v)$.
	The latter equality contradicts our initial assumption, namely~$v$ being updated in~$r_j$ to $r_j(v)$.
	
	Now, consider the case where~$j$ is even. Again, assume towards a contradiction that~$v \in V$ is updated in~$r_j$ to some value less than~$\frac{j}{2}$.
	Since~$j$ is even,~$r_j$ is the risk update of~$r_{j-1}$. Further, as $v$ is updated in $r_j$, Player~$1$ wins the game~$(\arena, \wincond \cap \safety(U))$ from $v$, where $U = \set{v' \in \dom(r_{j-1}) \mid r_{j-1}(v') \leq r_j(v)}$. Hence, he has a strategy~$\tau$ such that every play starting in~$v$  and consistent with~$\tau$ either violates~$\wincond$ or eventually visits some vertex~$v'$ with~$r_{j-1}(v') \leq r_j(v)$.
	We claim $r_{j-2}(v') = r_{j-1}(v')$ for all~$v' \in U$.
	
	Towards a contradiction, assume $r_{j-2}(v') \neq r_{j-1}(v')$ for some~$v' \in U$.
	Note that we have~$r_{j-1}(v') \leq r_j(v) < \frac{j}{2}$.
	Thus,~$v'$ is updated in~$r_{j-1}$ to some value strictly less than~$\frac{j}{2}$, which contradicts the induction hypothesis for~$j-1$.
	Hence, we indeed obtain $r_{j-2}(v') = r_{j-1}(v')$ for all~$v' \in U$.
	
	Thus, there are two types of vertices~$v'$ in~$U$:
	those for which~$r_{j-3}(v')$ is defined, which implies~$r_{j-3}(v') = r_{j-1}(v')$ due to the induction hypothesis and refinement, and those where~$r_{j-3}(v')$ is undefined, which implies $r_{j-2}(v') = r_{j-1}(v')$ due to the claim above.
	
	We claim that Player~$1$ wins $(\arena, \wincond \cap \safety(\set{v'' \in \dom(r_{j-3}) \mid r_{j-3}(v'') \leq r_j(v)}))$ from $v$, which implies $r_{j-2}(v) = r_j(v)$. This contradicts $v$ being updated in $r_j$, our initial assumption. 
	
	To this end, we construct a strategy~$\tau'$ from $v$ that either violates $\wincond$ or reaches a vertex~$v''$ with $r_{j-3}(v'')\le r_j(v)$ as follows. From $v$, $\tau'$ mimics $\tau$ until a vertex~$v'$ in $U$ is reached (if it is at all). If $v'$ is of the first type, then we have $r_{j-3}(v') = r_{j-1}(v') \le r_j(v)$. If $v'$ is of the second type, then $v'$ is updated in $r_{j-2}$ to some rank~$r_{j-2}(v') = r_{j-1}(v') \le r_j(v)$. As $r_{j-2}$ is the risk update of $r_{j-3}$, Player~$1$ has a strategy~$\tau_{v'}$ from $v'$ that either violates $\wincond$ or reaches a vertex~$v''$ with $r_{j-3}(v'') \le r_{j-2}(v') \le r_j(v)$. Thus, starting in $v'$, $\tau'$ mimics $\tau_{v'}$ from $v'$ until such a vertex is reached (if it is reached at all). Thus, every play that starts in $v$ and is consistent with $\tau'$ either violates $\wincond$ (as it has a suffix that does) or reaches a vertex~$v''$ with $r_{j-3}(v'') \le r_j(v)$, which proves our claim.\qed 
	\end{proof}
Lemma~\ref{lemma_ranktermination} implies that an algorithm computing the $r_j$ does not need to implement the definition of the two updates as presented above, but can be optimized by taking into account that a rank is never updated once set. However, for the proofs below, the definition presented above is more expedient, as it gives stronger preconditions to rely on, e.g., Lemma~\ref{lemma_disturbanceupdateproperties} and \ref{lemma_riskupdateproperties} only hold for the definition presented above.

Also, from the proof of Lemma~\ref{lemma_ranktermination}, we obtain an upper bound on the maximal rank of $r^*$. This in turn implies that the $r_j$ stabilize quickly, as $r_j = r_{j+1} = r_{j+2}$ implies $r_j = r^*$.
 
\begin{corollary}
\label{corollary_computationprops}
We have $\im(r^*) = \set{0, 1, \ldots, n}$ for some $n < \size{V}$ and $r^* = r_{2\size{V}}$.
\end{corollary}

The main result of this section shows that $r^*$ characterizes the resilience of vertices of finite resilience.

\begin{lemma}
\label{lemma_rank_correctness}
Let $r^*$ be defined for $\game$ as above, and let $v \in V$.
\begin{enumerate}
	\item\label{lemma_rank_correctness_finite} If $v \in \dom(r^*)$, then $r_\game(v) = r^*(v)$.
	\item\label{lemma_rank_correctness_infinite} If $v \notin \dom(r^*)$, then $r_\game(v) \in \set{\omega, \omega+1}$.
\end{enumerate}
\end{lemma}

\begin{proof}
\ref{lemma_rank_correctness_finite}.)
We show $r_\game(v) \leq r^*(v)$ and $r_\game(v) \geq r^*(v)$.

\textbf{\myquot{\boldmath$r_\game(v) \leq r^*(v)$}:} An $\alpha$-resilient strategy from $v$ is also $\alpha'$-resilient from $v$ for every $\alpha' \le \alpha$. Thus, to prove 
\begin{multline*}
r_\game(v) = \sup\{ \alpha\in\omega+2  \mid \\ \text{Player~$0$ has an $\alpha$-resilient strategy for $\game$ from $v$}\}  \leq r^*(v) 
\end{multline*}
we just have to show that Player~$0$ has no $(r^*(v)+1)$-resilient strategy from~$v$. By definition, for every strategy~$\sigma$ for Player~$0$, we have to show that there is a play~$\rho$ starting in $v$ and consistent with~$\sigma$ that has at most $r^*(v)$ disturbances and is winning for Player~$1$. So, fix an arbitrary strategy~$\sigma$.

 We  define a play with the desired properties by constructing longer and longer finite prefixes before finally appending an infinite suffix. During the construction, we ensure that each such prefix ends in $\dom(r^*)$ in order to be able to proceed with our construction.

The first prefix just contains the starting vertex~$(v,0)$, i.e., the prefix does indeed end in $\dom(r^*)$. Now, assume we have produced a prefix~$w(v',b')$ ending in some vertex~$v' \in \dom(r^*)$, which implies that $j_{v'}$ is defined. We consider three cases:

If $j_{v'} = 0$, then $v' \in \winreg_1(\game)$ by definition of $r_0$, i.e., Player~$1$ has a winning strategy~$\tau$ from $v$. Thus, we extend $w(v',b')$ by the unique disturbance-free play that starts in $v'$ and is consistent with $\sigma$ and $\tau$, without its first vertex. In that case, the construction of the infinite play is complete.

Second, if $j_{v'} > 0$ is odd, then $v'$ received its rank~$r^*(v')$ during a disturbance update. Hence, there is some $v''$ such that $(v',v'') \in D$ with $r^*(v') -1 = r^*(v'') $. In this case, we extend $w(v',b')$ by such a vertex~$v''$ to obtain the new prefix~$w(v',b')(v'',1)$, which satisfies the invariant, as $v''$ is in $\dom(r^*)$. Further, we have $j_{v''} < j_{v'}$ as the rank of $v''$ had to be defined in order to be considered during the disturbance update assigning a rank to $v'$.

Finally, if $j_{v'} > 0$ is even, then $v'$ received its rank~$r^*(v')$ during a risk update. We claim that Player~$1$ has a strategy~$\tau_{v'}$ that guarantees one of the following outcomes from $v'$: either the resulting play violates $\wincond$ or it encounters a vertex~$v''$ that satisfies $r^*(v'') \leq r^*(v')$ and $j_{v''} < j_{v'}$ (which implies $v'' \neq v'$). 		
	
	In that case, consider the unique disturbance-free play~$\rho'$ that starts in $v'$ and is consistent with $\sigma$ and the strategy~$\tau_{v'}$ as above. If $\rho'$ violates $\wincond$, then we extend $w(v',b')$ by $\rho'$ without its first vertex. In that case, the construction of the infinite play is complete.
	
	 If $\rho'$ does not violate $\wincond$, then we extend $w(v',b')$ by the prefix of $\rho'$ without its first vertex and up to (and including) the first occurrence of a vertex~$v''$ in $\rho'$ satisfying the properties described above. Note that this again satisfies the invariant.
	 
	It remains to argue our claim: $v'$ was assigned its  rank~$r^*(v') = r_{j_{v'}}(v')$ because it is in Player~$1$'s winning region in the game with winning condition~$\wincond \cap \safety(U)$, for 
	 \[U = \set{v'' \in \dom(r_{j_{v'}-1}) \mid r_{j_{v'}-1}(v'') \le r_{j_{v'}}(v')}.\]
	 	  Hence, from~$v'$, Player~$1$ has a strategy to either violate the winning condition or to reach~$U$. Thus, $r_{j_{v'}-1}(v'') = r^*(v'')$ for every $v'' \in \dom(r_{j_{v'}-1})$ yields $r^*(v'') \leq r^*(v')$. Finally, we have $j_{v''} < j_{v'}$, as the rank of $v'$ is assigned due to vertices in $U$ already having ranks. 

Note that only in two cases, we extend the prefix to an infinite play. In the other two cases, we just extend the prefix to a longer finite one. Thus, we first show that this construction always results in an infinite play. To this end, let $w_0(v_0,b_0)$ and $w_1 (v_1,b_1)$ be two of the prefixes constructed above such that $w_1(v_1,b_1)$ is an extension of $w_0(v_0,b_0)$. A simple induction proves $j_{v_1} < j_{v_0}$. Hence, as the value can only decrease finitely often, at some point an infinite suffix is added. Thus, we indeed construct an infinite play.

Finally, we have to show that the resulting play has the desired properties: by construction, the play starts in $v$ and is consistent with $\sigma$. Furthermore, by construction, it has a disturbance-free suffix that violates $\wincond$. Thus, by prefix-independence, the whole play also violates $\wincond$. It remains to show that it has at most $r^*(v)$ disturbances. To this end, let $w_0(v_0,b_0)$ and $w_1 (v_1,b_1)$ be two of the prefixes such that $w_1 (v_1,b_1)$ is obtained by extending $w_0(v_0,b_0)$ once. If the extension consists of taking the disturbance edge~$(v_0, v_1) \in D$, then we have $r^*(v_1) = r^*(v_0)+1$. The only other possibility is the extension consisting of a finite play prefix that is consistent with the strategy~$\tau_{v_0}$. Then, by construction, we obtain $r^*(v_1) \le r^*(v_0)$.
So, there are at most $r^*(v)$ many disturbances in the play, as the current rank decreases with every disturbance edge and does not increase with the other type of extension, but is always non-negative.

\textbf{\myquot{\boldmath$r_\game(v) \geq r^*(v)$}:}  Here, we construct a strategy~$\sigmaf$ for Player~$0$ that is $r^*(v)$-resilient from every $v \in \dom(r^*)$, i.e., from $v$, $\sigmaf$ has to be winning even under $r^*(v)-1$ disturbances. As every strategy is $0$-resilient, we only have to consider those~$v$ with $r^*(v) >0$.

The proof is based on the fact that $r^*$ is both stable under the disturbance and under the risk update, i.e., the disturbance update and the risk update of $r^*$ are $r^*$, which yields the following properties.
Let $(v,v') \in D$ be a disturbance edge such that $r^*(v) > 0$. Then, we have $r^*(v') \ge r^*(v) -1$. Also, for every $v \in \dom(r^*)$ with $r^*(v) > 0$, Player~$0$ has a winning strategy~$\sigma_v$ from $v$ for the game~$\game_v = (\arena, \wincond \cap \safety(\set{v' \in \dom(r^*) \mid r^*(v') < r^*(v)}))$ (note the strict inequality). Here, we apply determinacy of $\game_v$, as the risk update is formulated in terms of Player~$1$'s winning region.

Now, we define $\sigmaf$ to always mimic a strategy~$\sigma_{v_{\cur}}$ for some $v_\cur \in \dom(r^*)$, which is initialized by the starting vertex. The strategy~$\sigma_{v_{\cur}}$ is mimicked until a consequential (w.r.t.\ $\sigma_{v_\cur}$) disturbance edge is taken, say by reaching $v'$. In that case, the strategy~$\sigmaf$ discards the history of the play constructed so far, updates~$v_\cur$ to $v'$, and begins mimicking $\sigma_{v'}$. This is repeated ad infinitum. 

Now, consider a play that starts in $\dom(r^*)$, is consistent with $\sigmaf$, and has less than $r^*(v)$ disturbances. The part up to the first consequential disturbance edge (if it exists at all) is consistent with~$\sigma_v$. Now, let $(v_0, v_0')$ be the corresponding disturbance edge. Then, we have $r^*(v_0) \geq r^*(v)$, as~$\sigma_v$ being a winning strategy for the safety condition never visits vertices with a rank smaller than~$r^*(v)$. Thus, we conclude $r^*(v_0') \ge r^*(v_0) -1 \ge r^*(v) -1$. Similarly, the part between the first and the second consequential disturbance edge (if it exists at all) is consistent with $\sigma_{v_0'}$. Again, if $(v_1, v_1')$ is the corresponding disturbance edge, then we have $r^*(v_1') \ge r^*(v_1) -1 \ge r^*(v) - 2$. Continuing this reasoning shows that less than $r^*(v)$ (consequential) disturbance edges lead to a vertex~$v'$ with $r^*(v') > 0$, as the rank is decreased by at most one for every disturbance edge. The suffix starting in this vertex is disturbance-free and consistent with $\sigma_{v'}$. Hence, the suffix satisfies $\wincond$, i.e., by prefix-independence, the whole play satisfies $\wincond$ as well. Thus, $\sigmaf$ is indeed $r^*(v)$-resilient from every $v \in \dom(r^*)$.

\ref{lemma_rank_correctness_infinite}.)
Let $X = V \setminus \dom(r^*)$. The disturbance update of $r^*$ being $r^*$ implies that every disturbance edge starting in $X$ leads back to $X$. 
Similarly, the  risk update of $r^*$ being $r^*$ implies 
$X = \winreg_0(\game_X)$ for $\game_X = (\arena, \wincond \cap \safety(V \setminus X))$.
Thus, from every $v \in X$, Player~$0$ has a strategy~$\sigma_v$ such that every disturbance-free play that starts in $v$ and is consistent with $\sigma_v$ satisfies the winning condition~$\wincond$ and never leaves $X$. Using these properties, we construct a strategy~$\sigmaomega$ that is $\omega$-resilient from each~$v \in X$. Thus, $r_\game(v) \in \set{\omega, \omega+1}$.

The definition of the strategy~$\sigmaomega$ here is similar to the one above yielding the lower bound on the resilience. Again, $\sigmaomega$ always mimics a strategy~$\sigma_{v_{\cur}}$ for some $v_\cur \in X$, which is initialized by the starting vertex. The strategy~$\sigma_{v_{\cur}}$ is mimicked until a consequential (w.r.t.\ $\sigma_{v_\cur}$) disturbance edge is taken, say by reaching the vertex~$v'$. In that case, the strategy~$\sigmaomega$ discards the history of the play constructed so far, updates~$v_\cur$ to $v'$, and begins mimicking $\sigma_{v'}$. This is repeated ad infinitum. 

Due to the properties of the disturbance edges and the strategies~$\sigma_v$, such a play never leaves $X$, even if disturbances occur. Furthermore, if only finitely many disturbances occur, then the resulting play has a disturbance-free suffix that starts in some $v' \in X$ and is consistent with $\sigma_{v'}$. As $\sigma_{v'}$ is winning from $v'$ in $\game_X$, this suffix satisfies $\wincond$. Hence, by prefix-independence of $\wincond$, the whole play also satisfies $\wincond$. Thus, $\sigmaomega$ is indeed an $\omega$-resilient strategy from every $v \in X$. \qed 
\end{proof}

Combining Corollary~\ref{corollary_computationprops} and Lemma~\ref{lemma_rank_correctness}, we obtain an upper bound on the resilience of vertices with finite resilience.

\begin{corollary}
\label{corollary_resiliencedomain}
We have $r_\game(V) \cap \omega = \set{0, 1, \ldots, n}$ for some $n < \size{V}$.	
\end{corollary}

\subsection{Characterizing Vertices of Resilience $\omega+1$}
\label{subsec_resilomegaplus1}
Our goal in this subsection is to determine the vertices of resilience~$\omega+1$, i.e., those from which Player~$0$ can win even under an infinite number of disturbances. Intuitively, in this setting, we give Player~$1$ control over the disturbance edges, as he cannot execute more than infinitely many disturbances during a play.
 In the following, we prove this intuition to be correct. To this end, we transform the arena of the game so that at a vertex of Player~$0$, first Player~$1$ gets to chose whether he wants to take one of the disturbance edges and, if not, gives control to Player~$0$, who is then able to use a standard edge. 

Given a game~$\game = (\arena, \wincond)$ with $\arena = (V, V_0, V_1, E, D)$, we define the \emph{rigged game}~${\game_{\rig}} = (\arena', \wincond')$ with $\arena' = (V', V_0', V_1', E', D')$ such that $V' = V_0' \cup V_1'$ with $V_0' = \set{\overline{v} \mid v\in V_0}$ and $V_1' = V$, and $D' = \emptyset$. The set~$E'$ of edges  is the union of the following sets:
	\begin{itemize}
	\item $D$: Player~$1$ uses a disturbance edge.
	\item $ \set{ (v ,\overline{v}) \mid v \in V_0} $: Player~$1$ does not use a disturbance edge and yields control to Player~$0$.
	\item $\set{ (\overline{v},v') \mid (v,v') \in E \text{ and } v\in V_0 }$: Player~$0$ has control and picks a standard edge.
	\item $ \set{ (v ,v') \mid (v,v') \in E \text{ and } v\in V_1 }  $: Player~$1$ takes a standard edge.
\end{itemize}
Further, $\wincond' = \set{\rho \in (V')^\omega \mid h(\rho) \in \wincond}$ where $h$ is the homomorphism induced by $h(v )= v$ and $h(\overline{v}) = \epsilon$ for every $v \in V$. 

\begin{figure*}
	\centering
	\scalebox{.93}{
       \begin{tikzpicture}
               \ParityVertexOne{6}{v_6}{1}{(-1,0)}
               \ParityVertexOne{5?}{v_4}{1}{(-2.5,1)}
               \ParityVertexZero{5!}{\overline{v_4}}{1}{(-4,1)}
               \ParityVertexOne{4?}{v_5}{0}{(-2.5,-1)}
               \ParityVertexZero{4!}{\overline{v_5}}{0}{(-3.5,0)}
               \ParityVertexOne{3}{v_3}{1}{(-4,-1)}
               \ParityVertexOne{2?}{v_2}{2}{(-5,0)}
               \ParityVertexZero{2!}{\overline{v_2}}{2}{(-6,1)}
               \ParityVertexOne{1}{v_1}{1}{(-7.5,0)}
               \ParityVertexOne{8?}{v_7}{0}{(1,1)}
               \ParityVertexZero{8!}{\overline{v_7}}{0}{(2.5,1)}
               \ParityVertexOne{9}{v_8}{1}{(4,1)}
               \ParityVertexOne{10?}{v_{9}}{0}{(1,-1)}
               \ParityVertexZero{10!}{\overline{v_{9}}}{0}{(2.5,-1)}
               \ParityVertexOne{12}{v_{10}}{0}{(4,-1)}
               
               \path
                       (6) edge[bend right] (5?) edge[bend left] (4?) edge[bend left] (8?) edge[bend right] (10?)
                       (5!) edge [bend right=10] (5?) edge[bend right] (2?)
                       (5?) edge [bend right=10] (5!)
                       (4!) edge [bend right=10] (4?)
                       (4?) edge (3) edge [bend right=10] (4!)
                       (3) edge [bend left] (2?)
                       (2?) edge [bend right=10] (2!) edge (1)
                       (2!) edge [bend right=10] (2?)
                       (1) edge [loop above] (1)
                       (8?) edge [bend right=25] (9) edge [bend right=10] (8!)
                       (8!) edge [bend right=10] (8?)
                       (10?) edge [bend left=25] (12) edge [bend right=10] (10!)
                       (10!) edge [bend right=10] (10?)
                       (9) edge [bend right=25] (8?)
                       (12) edge [bend left=25] (10?);
                       
			\begin{pgfonlayer}{background}
				\coordinate (northwest) at (1 |- 2!);
				\coordinate (southwest) at (1 |- 3);
				\coordinate (W1-southmid) at ($(6) ! .5 ! (10?)$);
				\coordinate (W1-southeast) at ($(12) ! .5 ! (9)$);
				\coordinate (W0-northwest) at (W1-southmid |- W1-southeast);	
				\coordinate (southeast) at (12);
				\coordinate (northeast) at (9);
				
				\draw[black!15,thick,rounded corners,fill=black!5]
					($(southwest) + (-.5cm,-.6cm)$) -|
					($(W1-southmid) - (.1cm,0)$) |-
					($(W1-southeast) + (.75cm,.1cm)$) --
					($(northeast) + (.75cm,.6cm)$) --
					($(northwest) + (-.5cm,.6cm)$) --
					cycle;
					
				\node[anchor=north west] at ($(northwest) + (-.5cm,.6cm)$) {$\winreg_1$};
				
				\draw[black!15,thick,rounded corners,fill=black!5]
					($(W1-southmid) + (.1cm,0)$) |-
					($(W1-southeast) + (.75cm,-.1cm)$) --
					($(southeast) + (.75cm,-.6cm)$) -|
					cycle;
					
				\node[anchor=north west] at ($(W0-northwest) + (.1cm,-.1cm)$) {$\winreg_0$};
			\end{pgfonlayer}
       \end{tikzpicture}
       }
	      	\caption{The rigged game obtained for the game of Figure~\ref{fig:example}.} \label{fig:rigged_game} 
\end{figure*}
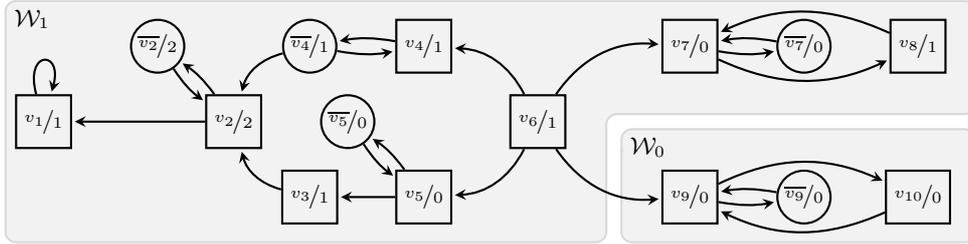

Figure~\ref{fig:rigged_game} illustrates the construction of a rigged game for the example game of Figure~\ref{fig:example} on Page~\pageref{fig:example} (note that the rigged game is also a parity game in this example). 
And indeed, the winning region of Player~$0$ corresponds to the vertices of resilience~$\omega + 1$ in the game of Figure~\ref{fig:example}.

The following lemma formalizes the observation that $\winreg_0({\game_{\rig}})$ characterizes the vertices of resilience~$\omega+1$ in $\game$. Note that we have no assumptions on $\game$ here. 

\begin{lemma}
\label{lemma_riggedcorrectness}
Let $v$ be a vertex of the  game~$\game$. Then, $v \in \winreg_0({\game_{\rig}})$ if and only if $r_\game(v) = \omega+1$.
\end{lemma}

\begin{proof}
The proof consists of constructing mappings between play prefixes and plays in both games, which are then used to transfer strategies between the games. This is conceptually straightforward, but technical due to the presence of the bits indicating whether a disturbance occurred or not. These have to be reconstructed to obtain proper mappings.

\textbf{\myquot{\boldmath$\Rightarrow$}:} Let Player~$0$ win ${\game_{\rig}}$ from $v $, say with winning strategy~$\sigma'$. We inductively translate play prefixes~$w$ in $\game$ into play prefixes~$t'(w)$ in ${\game_{\rig}}$ that satisfy the following invariant: $t'((v_0, b_0) \cdots (v_j, b_j))$ starts in $v_0 $ and ends in $v_j $.

For the induction start, we define $t'(v_0,b_0) = (v_0, 0)$; to define \[t'((v_0, b_0) \cdots (v_j, b_j)(v_{j+1}, b_{j+1})),\] we consider several cases:
\begin{itemize}
	\item If $b_{j+1} =1$, then $(v_j, v_{j+1}) \in D$, i.e., the play traverses the disturbance edge~$(v_j,v_{j+1})$. This move is mimicked by defining \[t'((v_0, b_0) \cdots (v_j, b_j)(v_{j+1}, b_{j+1})) = t'((v_0, b_0) \cdots (v_j, b_j)) \cdot (v_{j+1},0).\]
	
	\item If $b_{j+1} =0$, i.e., $(v_j, v_{j+1}) \in E$, and $v_{j} \in V_0$, then the play did not traverse a disturbance edge and instead allowed Player~$0$ to pick a standard edge~$(v_j,v_{j+1})$ to traverse. This move is mimicked by defining \[t'((v_0, b_0) \cdots (v_j, b_j)(v_{j+1}, b_{j+1})) = t'((v_0, b_0) \cdots (v_j, b_j)) \cdot (\overline{v_j},0) \cdot (v_{j+1},0).\]
	
	\item If $b_{j+1} =0$, i.e., $(v_j, v_{j+1}) \in E$, and $v_{j} \in V_1$, then the play traversed the standard edge~$(v_j,v_{j+1})$. This move is mimicked by defining \[t'((v_0, b_0) \cdots (v_j, b_j)(v_{j+1}, b_{j+1})) = t'((v_0, b_0) \cdots (v_j, b_j)) \cdot (v_{j+1},0).\]

\end{itemize}
Note that our invariant is satisfied in any case. Also, we lift $t'$ to infinite plays by taking limits as usual.

Let $d$ be the homomorphism induced by mapping $(v,b) \in V' \times\set{0,1}$ to $v \in V'$, i.e., $d$ removes the bits indicating the occurrence of disturbances.
Using the translation~$t'$, we define a strategy~$\sigma$ for Player~$0$ in $\game$ via 
\[\sigma(v_0 \cdots v_j) = \sigma'(d(t'((v_0, b_0) \cdots (v_j, b_j))) \cdot \overline{v_j} ),\]
 where $b_0 = 0$ and where for every $j'>0$, $b_{j'} =1$ if and only if $v_{j'} \neq \sigma(v_0 \cdots v_{j'-1})$, i.e., we reconstruct the consequential disturbances. A straightforward induction shows that for every play~$\rho = (v_0, b_0) (v_1, b_1) (v_2, b_2) \cdots$ in $\game$ that is consistent with $\sigma$, the play~$t'(\rho)$ is consistent with $\sigma'$. Hence, $t'(\rho) \in \wincond'$ for every $\rho$ starting in $v$. Furthermore, we have   $h(t'(\rho)) = v_0 v_1 v_2 \cdots \in \wincond$, as $t'(\rho) \in \wincond'$. Thus, $\rho = (v_0, b_0) (v_1, b_1) (v_2, b_2) \cdots$ is winning for Player~$0$.
As we have no restriction on the number of disturbances in $\rho$, $\sigma$ is $(\omega+1)$-resilient from $v$. Thus, $r_\game(v) = \omega+1$.

\textbf{\myquot{\boldmath$\Leftarrow$}:} Now, let $r_\game(v) = \omega+1$, i.e., Player~$0$ has an $(\omega+1)$-resilient strategy~$\sigma$ from $v$ in~$\game$. This time, we inductively define a translation~$t$ of play prefixes in ${\game_{\rig}}$ into play prefixes in $\game$. Here, it suffices to consider those prefixes that start and end in $V_1'$. For these, we satisfy the following invariant: if $w$ starts in $v_0 $ and ends in $v_j $, then $t(w) $ starts in $v_0$ and ends in $v_j$ as well. Note that $\game_\rig$ has no disturbance edges. Hence, the bits indicating whether such an edge has been traversed are always zero in plays of $\game_\rig$. Thus, we define $t(v_0,0 ) = (v_0,0)$ and consider several cases for the inductive step: 
\begin{itemize}
	\item First, assume we have a prefix of the form 
			$ (v_0,0)  \cdots (v_j, 0)  (v_{j+1},0) $
			for some $v_j \in V_0$, i.e., Player~$1$'s move simulates the disturbance edge $(v_j, v_{j+1}) \in D$. Then, we define 
\[		\quad\quad t((v_0,0)  \cdots (v_j, 0)  (v_{j+1},0)) =  t((v_0,0)  \cdots (v_j,0)  )\cdot (v_{j+1},1) \enspace .
\]	
	\item Next, assume we have a prefix of the form 
		$ (v_0,0)  \cdots (v_j,0)  (v_{j+1},0) $
		for some $v_j \in V_1$, i.e., Player~$1$'s move simulates the standard edge $(v_j, v_{j+1}) \in E$. Then, we define 
\[		\quad\quad t((v_0,0)  \cdots (v_j,0)  (v_{j+1},0)) =  t((v_0,0)  \cdots (v_j,0)  )\cdot (v_{j+1},0) \enspace .
\]
	\item Finally, the last case is a prefix of the form 
	$ (v_0,0)  \cdots (v_j,0)  (\overline{v_j}, 0) (v_{j+1},0) $ for some $v_j \in V_0$, i.e., Player~$0$'s move simulates the standard edge $(v_j, v_{j+1}) \in E$. Then, we define 
\[		\quad\quad t((v_0,0)  \cdots (v_j,0)  (\overline{v_j},0) (v_{j+1},0)) = t((v_0,0)  \cdots (v_j,0) )\cdot (v_{j+1},0) \enspace .
\]	
\end{itemize}
The invariant is satisfied in any case. Also, we can again lift $t$ to infinite plays via limits.

Now, let $d$ be the homomorphism induced by mapping $(v,b) \in V \times\set{0,1}$ to $v \in V$, i.e., $d$ again deletes the bits indicating the occurrence of disturbances.
Then, we define a strategy~$\sigma'$ for Player~$0$ in ${\game_{\rig}}$ via
\[\sigma'(v_0  \cdots v_j \overline{v_j}) = \sigma(d(t((v_0,0)  \cdots (v_j,0) ))).\]
A straightforward induction shows that for every play~$\rho$ that is consistent with $\sigma'$, the play $t(\rho)$ is consistent with $\sigma$. Hence, if $\rho$ starts in $v $, then $t(\rho)$ satisfies the winning condition, as $\sigma$ is $(\omega+1)$-resilient from $v$. Let $t(\rho) = (v_0, b_0)(v_1, b_1)(v_2, b_2) \cdots $. Then, $v_0 v_1 v_2 \cdots \in \wincond$. Now, $h(\rho) = v_0 v_1 v_2 \cdots$ implies $\rho \in \wincond'$. Thus, $\sigma'$ is a winning strategy for Player~$0$ from $v $.\qed 
\end{proof}

With an adaption of the rigged game, one can also directly characterize the vertices with resilience~$\omega$. However, since our algorithm and the rigged game already provide   an indirect characterization, we do not present this construction here.


Furthermore, the proof of Lemma~\ref{lemma_riggedcorrectness}
 also yields the preservation of positional and finite-state strategies.
To this end, consider the first implication proved above. If $\sigma$ is positional (finite-state), then $\sigma'$ is positional (finite-state) as well. Thus, applying both implications yields the following corollary. 

\begin{corollary}
\label{corollary_rigged} Let $\game$ and $\game_{\rig}$ be defined as above and $v$ a vertex of $\game$.
\begin{enumerate}
	\item\label{corollary_rigged_positional}
 Assume Player~$0$ has a positional winning strategy for ${\game_{\rig}}$ from $v $. Then, Player~$0$ has an $(\omega+1)$-resilient  positional strategy for $\game$ from $v$. 
\item \label{corollary_rigged_finitestate}
Assume Player~$0$ has a finite-state winning strategy for ${\game_{\rig}}$ from $v $. Then, Player~$0$ has an $(\omega+1)$-resilient  finite-state strategy (of the same size) for $\game$ from $v$. 
\end{enumerate}
\end{corollary}

\subsection{Computing Optimally Resilient Strategies}
\label{subsec_strategies}
This subsection is concerned with computing the resilience and optimally resilient strategies. Here, we focus on  positional and finite-state strategies, which are sufficient for the majority of winning conditions in the literature. Nevertheless, it is easy to see that our framework is also applicable to infinite-state strategies. 
 
In the proof of Lemma~\ref{lemma_rank_correctness}, we construct strategies~$\sigmaf$ and $\sigmaomega$ such that $\sigmaf$ is $r_\game(v)$-resilient from every $v$ with $r_\game(v) \in \omega$ and such that $\sigmaomega$ is~$\omega$-resilient from every~$v$ with $r_\game(v) \geq \omega$. 
Both strategies are obtained by combining winning strategies for some game~$(\arena, \wincond \cap \safety(U))$. However, even if these winning strategies are positional, the strategies~$\sigmaf$ and $\sigmaomega$ are in general not positional. Nonetheless, we show in the proof of Theorem~\ref{theorem_main} that such positional winning strategies and a positional one for $\game_\rig$ can be combined into a single positional optimally resilient strategy. 

Recall the requirements from Subsection~\ref{subsec_finiteresil} for a game $(\arena, \wincond)$: $\wincond$ is prefix-independent and the game~$\game_U$ is determined for every $U \subseteq V$, where we write $\game_U$ for the game~$(\arena, \wincond \cap \safety(U))$ for some $U \subseteq V$. 
 To prove the results of this subsection, we need to impose some additional effectiveness requirements: we require that each game~$\game_U$ and the rigged game~${\game_{\rig}}$ can be effectively solved. Also, we first assume that Player~$0$ has positional winning strategies for each of these games, which have to be effectively computable as well. We discuss the severity of these requirements in Section~\ref{sec_discussion}.

\begin{theorem}
\label{theorem_main}
Let $\game$ satisfy all the above requirements. Then, the resilience of $\game$'s vertices and a positional optimally resilient strategy can be effectively computed.
\end{theorem}

To prove this result, we refine the following standard technique that combines positional winning strategies for  games with prefix-independent winning conditions.

Assume we have a positional strategy $\sigma_v$ for every vertex~$v$ in some set~$W \subseteq V$ such that $\sigma_v$ is winning from~$v$. Furthermore, let $R_v$ be the set of vertices visited by plays that start in $v$ and are consistent with $\sigma_v$. Also, let $m(v) = \min_\prec\set{v' \in V \mid v \in R_{v'}}$ for some strict total ordering~$\prec$ of $W$.
  Then, the positional strategy~$\sigma$ defined by $\sigma(v) = \sigma_{m(v)}(v)$ is winning from each $v \in W$, as along every play that starts in some $v \in W$ and is consistent with $\sigma$, the value of the function~$m$ cannot increase. Thus, after it has stabilized, the remaining suffix is consistent with some strategy~$\sigma_{v'}$. Hence, the suffix is winning for Player~$0$ and prefix-independence implies that the whole play is winning for her as well.
  
  Here, we have to adapt this reasoning to respect the resilience of the vertices and to handle disturbance edges. Also, we have to pay attention to vertices of resilience~$\omega+1$, as plays starting in such vertices have to be winning under infinitely many disturbances. 

\begin{proof}[of Theorem~\ref{theorem_main}]
The effective computability of the resilience follows from the effectiveness requirements on $\game$: to compute the ranking~$r^*$, it suffices to compute the disturbance and risk updates. The former are trivially effective while the effectiveness of the latter ones follows from our assumption. Lemma~\ref{lemma_rank_correctness} shows that $r^*$ correctly determines the resilience of all vertices with finite resilience. Finally, by solving the rigged game, we also determine the resilience of the remaining vertices (Lemma~\ref{lemma_riggedcorrectness}). Again, this game can be solved due to our assumption. 

Thus, it remains to show how to compute a positional optimally resilient strategy. To this end, we compute a positional strategy~$\sigma_v$ for every $v$ satisfying the following:
\begin{itemize}
	\item For every $v \in V$ with $r_\game(v) \in \omega \setminus \set{0}$, the strategy~$\sigma_v$ is winning for Player~$0$ from $v$ for the game~$(\arena, \wincond \cap \safety(\set{v' \in V \mid r_\game(v') < r_\game(v)}))$. We have shown the existence of such a strategy in the proof of Item~\ref{lemma_rank_correctness_finite} of Lemma~\ref{lemma_rank_correctness}.
	
	\item For every $v \in V$ with $r_\game(v)=\omega$, the strategy~$\sigma_v$ is winning for Player~$0$ from $v$ for the game~$(\arena, \wincond \cap \safety(\set{v' \in V \mid r_\game(v') \in \omega}))$.
	 We have shown the existence of such a strategy in the proof of Item~\ref{lemma_rank_correctness_infinite} of Lemma~\ref{lemma_rank_correctness}.

	\item For every $v \in V$  with $r_\game(v) = \omega+1$, the strategy~$\sigma_v$ is~$(\omega+1)$-resilient from $v$. 
	The existence of such a strategy follows from Item~\ref{corollary_rigged_positional} of Corollary~\ref{corollary_rigged}, as we assume Player~$0$ to win $\game_\rig$ with positional strategies.
	
	\item For every $v \in V$ with $r_\game(v) = 0$, we fix an arbitrary positional strategy~$\sigma_v$ for Player~$0$. 

\end{itemize}

Furthermore, we fix a strict linear order~$\prec$ on $V$ such that $v \prec v'$ implies $r_\game(v) \le r_\game(v')$, i.e., we order the vertices by ascending resilience. For $v\in V$ with $r_\game(v) \neq \omega+1$, let $R_v$ be the set of vertices reachable via disturbance-free plays that start in $v$ and are consistent with~$\sigma_v$. On the other hand, for $v \in V$ with $r_\game(v) = \omega+1$, let $R_v$ be the set of vertices reachable via plays with arbitrarily many disturbances that start in $v$ and are consistent with~$\sigma_v$.

We claim $R_v \subseteq \set{v' \in V \mid r_\game(v') \ge r_\game(v)}$ for every $v \in V$ ($\ast$). For $v$ with $r_\game(v) \neq \omega +1$ this follows immediately from the choice of $\sigma_v$. Thus, let $v$ with $r_\game(v) = \omega+1$. Assume $\sigma_v$ reaches a vertex~$v'$ of resilience~$r_\game(v') \neq \omega+1$. Then, there exists a play~$\rho'$ starting in $v'$ that is consistent with $\sigma_v$, has less than $\omega+1$ many disturbances and is losing for Player~$0$. Thus, the play obtained by first taking the play prefix to $v'$ and then appending $\rho'$ without its first vertex yields a play starting in $v$, consistent with $\sigma_v$, but losing for Player~$0$. This play witnesses that $\sigma_v$ is not $(\omega+1)$-resilient from $v$, which contradicts our assumption and thus concludes the proof of the claim for the case~$r_\game(v) = \omega+ 1$.

Let $m \colon V \rightarrow V$ be given as $m(v) = \min_\prec\set{v' \in V \mid v \in R_{v'}}$ and define the positional strategy~$\sigma$ as $\sigma(v) = \sigma_{m(v)}(v)$. By our assumptions, $\sigma$ can be effectively computed. It remains to show that it is optimally resilient.

To this end, we apply the following two properties of edges~$(v,v')$ that may appear during a play that is consistent with $\sigma$, i.e., we either have $v \in V_0 $ and $\sigma(v) = v'$ (which implies $(v,v') \in E$), or $v \in V_1$ and $(v,v') \in E$, or $v \in V_0$ and $(v,v') \in D$: 
\begin{enumerate}
	\item If $(v,v') \in E$, then we have $r_\game(v) \le r_\game(v')$ and $m(v) \ge m(v')$. The first property follows from minimality of $m(v)$ and ($\ast$) while the second follows from the definition of $R_v$. 

	\item If $(v,v') \in D$, then we distinguish several subcases, which all follow immediately from the definition of resilience:
	\begin{itemize}
		\item If $r_\game(v) \in \omega$, then $r_\game(v') \ge r_\game(v) -1$. 
		\item If $r_\game(v) = \omega$, then $r_\game(v') = \omega$, and
		\item If $r_\game(v) = \omega + 1$, then $r_\game(v') = \omega + 1$ and $m(v) \ge m(v')$ (here, the second property follows from the definition of $R_v$ for~$v$ with~$r_\game(v) = \omega + 1$, which takes disturbance edges into account).
	\end{itemize}
\end{enumerate}

Now, consider a play~$\rho = (v_0, b_0) (v_1, b_1) (v_2, b_2) \cdots$ that is consistent with $\sigma$. If $r_\game(v_0) = 0$ then we have nothing to show, as every strategy is $0$-resilient from $v$.

Now, assume $r_\game(v_0) \in \omega \setminus\set{0}$. We have to show that if $\rho$ has less than $r_\game(v_0)$ disturbances, then it is winning for Player~$0$. An inductive application of the above properties shows that in that case the last disturbance edge leads to a vertex of non-zero resilience. Furthermore, as the values~$m(v_j)$ are only decreasing afterwards, they have to stabilize at some later point.
Hence, there is some suffix of~$\rho$ that starts in some~$v'$ with non-zero resilience and that is consistent with the strategy~$\sigma_{v'}$.
Thus, the suffix is winning for Player~$0$ by the choice of $\sigma_{v'}$ and prefix-independence implies that $\rho$ is winning for her as well.

Next, assume $r_\game(v_0) = \omega$. We have to show that if $\rho$ has a finite number of disturbances, then it is winning for Player~$0$. Again, an inductive application of the above properties shows that in that case the last disturbance edge leads to a vertex of resilience~$\omega$ or~$\omega+1$. Afterwards, the values~$m(v_j)$ stabilize again.
Hence, there is some suffix of~$\rho$ that starts in some~$v'$ with non-zero resilience and that is consistent with the strategy~$\sigma_{v'}$.
Thus, the suffix is winning for Player~$0$ by the choice of $\sigma_{v'}$ and prefix-independence implies that $\rho$ is winning for her as well.

Finally, assume $r_\game(v_0) = \omega+1$. Then, the above properties imply that $\rho$ only visits vertices with resilience~$\omega+1$ and that the values~$m(v_j)$ eventually stabilize. Hence, there is a suffix of $\rho$ that is consistent with some $(\omega+1)$-resilient strategy~$\sigma_{v'}$, where $v'$ is the first vertex of the suffix. Hence, the suffix is winning for Player~$0$, no matter how many disturbances occur. This again implies that $\rho$ is winning for her as well. \qed
\end{proof}

The algorithm determining the vertices' resilience and a positional optimally resilient strategy first computes~$r^*$ and the winner of the rigged game. This yields the resilience of $\game$'s vertices. Furthermore, the strategy is obtained by combining winning strategies for the games~$\game_U$ and for the rigged game as explained above.

Next, we analyze the complexity of the algorithm sketched above in some more detail.
The inductive definition of the $r_j$ can be turned into an algorithm computing~$r^*$ (using the results of Lemma~\ref{lemma_ranktermination} to optimize the naive implementation), which has to solve $\bigo(\size{V})$ many games (and compute winning strategies for some of them) with winning condition~$\wincond \cap \safety(U)$.
Furthermore, the rigged game, which is of size~$\bigo(\size{V})$, has to be solved and winning strategies have to be determined.
Thus, the overall complexity is in general dominated by the complexity of solving these tasks.

We explicitly state one complexity result for the important case of parity games, using the fact that each of these games is then a parity game as well. Also, we use a quasipolynomial time algorithm for solving parity games~\cite{CJKLS16,FearnleyJS0W17,JurdzinskiL17,Lehtinen18} to solve the games $\game_U$ and ${\game_{\rig}}$.

\begin{theorem}
\label{theorem_main_parity}
Optimally resilient strategies in parity games are positional and can be computed in quasipolynomial time.
\end{theorem}

Using similar arguments, one can also analyze games where positional strategies do not suffice.
As above, assume $\game$ satisfies the same assumptions on determinacy and effectiveness, but only require that Player~$0$ has finite-state winning strategies for each game with winning condition $(\arena, \wincond \cap \safety(U))$ and for the rigged game~${\game_{\rig}}$. Then, one can show that she has a finite-state optimally resilient strategy.
In fact, by reusing memory states, one can construct an optimally resilient strategy that it is not larger than any constituent strategy.

\section{Discussion}
\label{sec_discussion}
In this section, we discuss the assumptions required to be able to compute positional (finite-state) optimally resilient strategies with the algorithm presented in Section~\ref{sec_results}. Here, we only consider the case of positional strategies. The case of finite-state strategies is analogous. 

To this end, fix a game~$\game = (\arena, \wincond)$ with vertex set~$V$ and recall that ${\game_{\rig}}$ is the corresponding rigged game and that we defined $\game_U = (\arena, \wincond \cap \safety(U))$ for $U \subseteq V$.  Now, the assumptions on~$\game$ that need to be satisfied for Theorem~\ref{theorem_main} to hold are as follows:
\begin{enumerate}
	\item The game~$\game_U$ is determined for every $U \subseteq V$.
	\item Player~$0$ has a positional winning strategy from every vertex in her winning regions in the~$\game_U$ and in the game~${\game_{\rig}}$.
	\item Each~$\game_U$ and the game~${\game_{\rig}}$ can be effectively solved and positional winning strategies can be effectively computed for each such game.  
	\item $\wincond$ is prefix-independent.
\end{enumerate}

First, consider the determinacy assumption. For $W \subseteq V$ let $\arena \setminus W$ denote the arena obtained from $\arena$ by removing all vertices from $W$, as well as all edges from or to vertices in $W$. It is easy to show that $\arena \setminus W$ has no terminal vertices, if $W$ is the winning region of Player~$1$ in a safety game played in $\arena$. 
Now, it is straightforward to show
\[
\winreg_0(\game_U) = \winreg_0(\arena \setminus W, \wincond \cap (V \setminus W)^\omega)
\]
and
\[
\winreg_1(\game_U) = W \cup \winreg_1(\arena \setminus W, \wincond \cap (V \setminus W)^\omega)
\]
where $W = \winreg_1(\arena,\safety(U))$. 
Thus, one can remove the winning region of Player~$1$ in the safety game and then consider the subgame of $\game$ played in Player~$0$'s winning region of the safety game. Thus, all subgames of $\game$ being determined suffices for the determinacy requirement being satisfied. The winning conditions one typically studies, e.g., parity and in fact all Borel ones~\cite{Martin75}, satisfy this property.

The next requirement concerns the existence of positional winning strategies for the games~$\game_U$ and ${\game_{\rig}}$. For the~$\game_U$, this requirement is satisfied if Player~$0$ has positional winning strategies for all subgames of $\game$, as argued above. As every positional  optimally resilient strategy is also a winning strategy in a certain subgame, this condition is necessary. Now, consider~$\game_\rig$, whose winning condition can be written as $h^{-1}(\wincond)$ for the homomorphism~$h$ from Subsection~\ref{subsec_resilomegaplus1}.
 The winning conditions one typically studies, e.g., the Borel ones, are closed w.r.t. such supersequences. If $\game$ is from a class of winning conditions that allows for positional  winning strategies for Player~$0$, then this class typically also contains ${\game_{\rig}}$.
Also, the assumption on the effective solvability and computability of positional strategies is obviously necessary, as we solve a more general problem when determining optimally resilient strategies.

Finally, let us consider prefix-independence.
If the winning condition is not prefix-independent, then the algorithm presented in Section~\ref{sec_results} does not compute the resilience of vertices correctly.
In fact, recall that a winning condition~$\wincond$ is prefix-independent if, for all plays~$\rho$ and all play prefixes~$w$, we have~$\rho \in \wincond$ if and only if~$w\rho \in \wincond$.
We show that neither implication of this equivalence suffices on its own for the algorithm from Section~\ref{sec_results} to compute the correct resilience of vertices.

\begin{figure}[h]
	\centering
	\begin{tikzpicture}
		\begin{scope}
			\node[p0] (0) at (0,0) {$v$};
			\node[p0] (1) at (2,0) {$v'$};
			
			\node[anchor=west] at (3.75,0) {$\wincond_k = \{v_0v_1v_2\cdots \in V^\omega \mid\card{\set{j \mid v_j = v}} \leq k\}$};

			\path
				(0) edge[bend right=10] (1)
				(1) edge[loop right] (1) edge[fault,bend right=10] (0);
				
			\begin{pgfonlayer}{background}
			\draw[black!15,thick,rounded corners,fill=black!5]
				($(0) - (1.5cm,.45cm)$) |-
				($(1) + (1cm,.45cm)$) |-
				cycle;
				
			\node[anchor=north west] at ($(0) - (1.5cm,-.45cm)$) {$\winreg_0$};
			\end{pgfonlayer}
		\end{scope}
		
	\end{tikzpicture}	
	\caption{Counterexample for requiring only the implication from right to left from the definition of prefix-independence for the computation of resilience.}
	\label{fig:prefix-dependent}
\end{figure}

First, consider the family~$\game_k= (\arena, \wincond_k)$ of games shown in Figure~\ref{fig:prefix-dependent}.
In $\game_k$, it is the goal of Player~$0$ to avoid more than~$k$ visits to~$v$.
Hence, for all plays~$\rho$ and all play prefixes~$w$ we have that~$w\rho \in \wincond$ implies~$\rho \in \wincond$.

In each of the~$\game_k$, a visit to~$v$ only occurs via a disturbance or if the initial vertex is $v$.
Hence, we have~$r_{\game_k}(v) = k$ and~$r_{\game_k}(v') = k+1$.
If we apply the algorithm from Section~\ref{sec_results}, however, the initial ranking function~$r_0$ has an empty domain, since we have $\winreg_1(\game_k) = \emptyset$.
Thus, the computation of the~$r_j$ immediately stabilizes, yielding~$r^*$ with empty domain.
Hence, that algorithm does, in general, not compute the correct resilience if only the implication from right to left from the definition of prefix-independence is satisfied.

\begin{figure}
\centering
\begin{tikzpicture}
	\node[p0] (0) at (0,0) {$v$};
	\node[p0] (1) at (2,0) {$v'$};
	
	\path
		(0) edge[loop left] (0)
		(0) edge[fault] (1)
		(1) edge[loop right] (1);
		
	\node[anchor=west] at (3.75,0) {$\wincond = V^\omega\setminus\{(v')^\omega\}$};
	
	\begin{pgfonlayer}{background}	
		\draw[black!15,thick,rounded corners,fill=black!5]
		($(0) - (1.5cm,.45cm)$) |-
		($(0) + (.75cm,.45cm)$) |-
		cycle;
				
		\node[anchor=north west] at ($(0) - (1.5cm,-.45cm)$) {$\winreg_0$};
	\end{pgfonlayer}
	
	\begin{pgfonlayer}{background}	
		\draw[black!15,thick,rounded corners,fill=black!5]
		($(1) - (.75cm,.45cm)$) |-
		($(1) + (1.5cm,.45cm)$) |-
		cycle;
				
		\node[anchor=north east] at ($(1) + (1.5cm,.45cm)$) {$\winreg_1$};
	\end{pgfonlayer}
\end{tikzpicture}
\caption{Counterexample for requiring only the implication from left to right from the definition of prefix-independence for the computation of resilience.}
\label{fig:suffix-closed}
\end{figure}

Conversely, consider the game~$\game$ shown in Figure~\ref{fig:suffix-closed}.
The winning condition of this game satisfies that, for all play prefixes~$w$ and all plays~$\rho$, we have that~$\rho \in \wincond$ implies~$w\rho \in \wincond$.
If we apply the algorithm from Section~\ref{sec_results}, however, the initial ranking~$r_0$ has the domain~$\set{v'}$ with~$r_0(v') = 0$, due to~$\winreg_1(\game) = \set{v'}$.
The disturbance update of~$r_0$ then yields the ranking~$r_1$ with~$r_1(v) = 1$ due to the single disturbance edge of~$\game$ and with~$r_1(v') = 0$.
At this point, the rankings stabilize and we obtain~$r^* = r_1$.

While we indeed have~$r_{\game}(v') = 0 = r^*(v')$, we furthermore have~$r_{\game}(v) = \omega + 1 \neq r^*(v)$, as every play starting in vertex~$v$ satisfies the winning condition.
Hence, this example showcases that the implication from left to right from the definition of prefix-independence also does not suffice for the algorithm from Section~\ref{sec_results} to correctly compute the resilience.
Thus, we indeed require full prefix-independence of the winning condition as a precondition for the correctness of that algorithm.

In the following subsection, we show that one can still leverage our algorithm from Section~\ref{sec_results} in order to compute the resilience of a wide range of games with prefix-dependent winning conditions.
To this end, we extend the framework of game reductions to games with disturbances, in such a way that the existence of $\alpha$-resilient strategies is preserved. Using this framework shows that Player~$0$ has a finite-state optimally resilient strategy in every game with $\omega$-regular winning condition.

\subsection{Prefix-dependent Winning Conditions}
\label{subsec_prefixdep}

We begin by introducing some notation regarding game reductions. 
An arena $\arena = (V, V_0, V_1, E, D)$ and a memory structure~$\mem = (M, \init, \update)$ for $\arena$ induce the expanded arena $\arena\times\mem = (V \times M, V_0 \times M, V_1 \times M, E', D')$ where~$E'$ is defined via $((v,m), (v',m')) \in E'$ if and only if $(v,v') \in E$ and $\update(m, v' ) = m'$. The disturbance edges~$D'$ are defined analogously, i.e.,  $((v,m), (v',m')) \in D'$ if and only if $(v,v') \in D$ and $\update(m, v' ) = m'$.
Every play $(v_0, b_0) (v_1,b_1) (v_2,b_2) \cdots$ in $\arena$ has a unique extended play 
	\[ \ext(\rho) = ((v_0, m_0), b_0) ((v_1, m_1),b_1) ((v_2, m_2),b_2) \cdots \] 
	in $\arena \times \mem$ defined by $m_0 = \init(v_0)$ and $m_{j+1} = \update(m_j,  v_{j+1})$, i.e., $m_j = \update^+(v_0 \cdots v_j)$. Play prefixes are translated analogously.

\begin{remark}
Let $\rho$ be a play in $\game$. Then, $\disturbances(\rho) = \disturbances(\ext(\rho))$.
\end{remark}

A game $\game = (\arena, \wincond)$ is {reducible} to $\game' = (\arena', \wincond')$
via $\mem$, written $\game \le_{ \mem } \game'$, if $\arena' = \arena \times
\mem$ and every play $\rho$ in $\game$ is won by the same player that wins $\ext(\rho)$ in $\game'$.

\begin{lemma} 
\label{lemma_reductionlemma} 
Let $\game \le_{
\mem } \game'$. Then, $r_\game(v) = r_{\game'}(v,\init(v))$ for all vertices~$v$ of $\game$.
\end{lemma}

\begin{proof}
We show that Player~$0$ has an $\alpha$-resilient strategy~$\sigma'$ for $\game'$
from $(v, \init(v))$ if and only if she has an $\alpha$-resilient strategy~$\sigma$ for~$\game$ from $v$, which implies our claim. The translation of the strategies is the same as in the disturbance-free setting (see, e.g.,~\cite{Kaiser11}), but here we have to argue about resilience instead of just winning. 

\textbf{\myquot{\boldmath$\Leftarrow$}:} Given a strategy~$\sigma$ for $\game$, we define $\sigma'$ for $\game'$ via 
	\[ \sigma'((v_0, m_0)\cdots(v_j, m_j)) = \sigma(v_0 \cdots v_j) \enspace . \] Consider a play~$\rho' = ((v_0, m_0),b_0)((v_1, m_1),b_1)((v_2, m_2),b_2) \cdots$ consistent with $\sigma'$. If $m_0 = \init(v_0)$, then $\rho' = \ext(\rho)$ for $\rho = (v_0,b_0)( v_1,b_1)( v_2,b_2) \cdots$, which is consistent with $\sigma$. Hence, $\rho'$ and $\rho$ have the same winner and the same number of disturbances. Hence, if $\sigma$ is $\alpha$-resilient from a vertex~$v$, then $\sigma'$ is $\alpha$-resilient from $(v, \init(v))$.

\textbf{\myquot{\boldmath$\Rightarrow$}:} Given a strategy~$\sigma'$ for $\game'$, we define $\sigma$ for $\game$ via $\sigma(v_0 \cdots v_j) = v$, if $\sigma'((v_0, m_0)\cdots(v_j,m_j)) = (v,m)$ for some $m \in M$, where $m_{j'} = \update^+(v_0 \cdots v_{j'})$. 

A straightforward induction shows that a play in $\game$ is consistent with $\sigma$ if and only if its extended play in $\game'$ is consistent with $\sigma'$. Thus, these plays have the same winner and the same number of disturbances. Thus, again, if $\sigma'$ is $\alpha$-resilient from a vertex~$(v, \init(v))$ then $\sigma$ is $\alpha$-resilient from $v$.\qed 
\end{proof}

As usual for game reductions, we obtain a finite-state strategy for $\game$ when starting with a positional strategy in $\game'$. To this end, consider the proof of the second implication above. If~$\sigma$ is positional, then the strategy~$\sigma'$ is implemented by $\mem$ and the next-move function~$\nxt$ given by $\nxt(v,m) = v'$, if $\sigma(v,m) = (v',m')$ for some $m' \in M$.  

A similar construction works in case $\sigma'$ is finite-state, say implemented by $\mem'$. Then,~$\sigma$ is implemented by the product of $\mem$ and $\mem'$, which is defined as expected (we refer to, e.g.,~\cite{Kaiser11} for a formal definition). Altogether, we obtain the following result.

\begin{corollary}
\label{corollary_finitestate}
Let $\game \le_\mem \game'$.
\begin{enumerate}
	\item If Player~$0$ has an $\alpha$-resilient positional strategy from $(v, \init(v))$ in $\game'$, then she has an $\alpha$-resilient finite-state strategy from $v$ in $\game$, which is implemented by $\mem$. 
\item If Player~$0$ has an $\alpha$-resilient finite-state strategy from $(v, \init(v))$ in $\game'$, say implemented by $\mem'$, then she has an $\alpha$-resilient finite-state strategy from $v$ in $\game$, which is implemented by the product of $\mem$ and $\mem'$. 
\end{enumerate}
\end{corollary}

Now, we can formulate the main theorem of this subsection, which shows that prefix-dependence is not a restriction, as long as the game is reducible to a prefix-independent one. Note that this is in particular true for every $\omega$-regular winning condition (see, e.g.,~\cite{GraedelThomasWilke02}): every such condition is recognized by a deterministic parity automaton, which can be turned into a memory structure which allows to reduce the original game to a parity game. 

\begin{theorem}
\label{thm_reductions}
Let $\game \le_\mem \game'$ so that $\game'$ has a prefix-independent winning condition, can be effectively computed from $\game$, and satisfies the assumptions from Section~\ref{subsec_strategies} (with finite-state strategies).

Then, the resilience of $\game$'s vertices and a finite-state optimally resilient strategy can be effectively computed. 
\end{theorem}

\begin{proof}
This is a direct consequence of Lemma~\ref{lemma_reductionlemma} and Theorem~\ref{theorem_main}. To obtain an optimally resilient strategy, we apply Corollary~\ref{corollary_finitestate} for finite-state strategies.\qed 
\end{proof}

Recall the family of games shown in Figure~\ref{fig:prefix-dependent} in which Player~$0$ aims to prevent more than~$k$ visits to the vertex~$v_1$ for some parameter~$k \in \omega$.
Such a game can be reduced to a parity game using a memory structure implementing a counter up to~$k+1$. Such a memory structure has $k+1$ memory states, and a straightforward pumping argument shows that there is no smaller memory structure.

Thus, we obtain an optimally resilient strategy for Player~$0$ that is implemented by a memory structure with~$k+1$ states.
While this strategy is indeed optimally resilient, it is not of minimal size: in fact, the unique strategy for Player~$0$ in $\game_k$ is positional and optimally resilient.
Thus, the approach of computing optimally resilient strategies for games with prefix-dependent winning conditions via reductions to prefix-independent winning conditions is not optimal in that sense, as it may yield unnecessarily large optimally resilient strategies.
In current research, we study how to synthesize minimal optimally resilient strategies for games with prefix-dependent winning conditions.

Moreover, in the case of prefix-dependent winning conditions, the question arises whether or not optimally resilient strategies may be necessarily larger than winning ones.
It is easy to construct a game in which Player~$0$ has a positional winning strategy, but an optimally resilient one requires an infinite amount of memory. One example is a game with a dedicated vertex~$v$ with a self-loop, such that using the self-loop ad infinitum is winning for Player~$0$. Furthermore, there is a disturbance edge leading from $v$ into a disturbance-free subgame in which Player~$0$ needs an infinite amount of memory to win. 

However, this example is not very useful, as Player~$0$ needs infinite memory to win the game from some vertex of her winning region. A more interesting question for further research is whether a result similar to Theorem~\ref{theorem_main} holds true for prefix-dependent games with positional winning strategies, e.g., weak parity games~\cite{Chatterjee08} or bounded parity games~\cite{ChatterjeeHenzingerHorn09}. However, for both of these conditions, monotonicity arguments allow to transform finite-state optimally resilient strategies into positional ones (similar to the construction in~\cite[Section 5]{FZ14}). However, these arguments rely on monotonicity properties of the parity condition and are therefore unlikely to be generalizable. On the other hand, we are not aware of an example of a class of winning conditions that always allow for positional winning strategies for Player~$0$, but require memory to implement optimally resilient strategies. In future work, we investigate whether the blowup introduced by the reduction can be avoided.

\section{Outlook}
\label{sec_outlook}
We have developed a fine-grained view on the quality of strategies:
instead of evaluating whether or not a strategy is winning, we compute its resilience against intermittent disturbances.
While this measure of quality allows constructing ``better'' strategies than the distinction between winning and losing strategies, there remain aspects of optimality that are not captured in our notion of resilience.
In this section we discuss these aspects and give examples of games in which there are crucial differences between optimally resilient strategies.
In further research, we aim to synthesize optimal strategies with respect to these criteria.

As a first example, consider the parity game shown in Figure~\ref{fig:recovery}.
Vertices~$v_0$ and~$v_3$ have resilience~$1$ and~$\omega+1$, respectively, while vertices~$v_1$,~$v_2$, and~$v'_2$ have resilience~$0$.
Player~$0$'s only choice consists of moving to~$v_2$ or to~$v'_2$ from~$v_1$.
Let~$\sigma$ and~$\sigma'$ be strategies for Player~$0$ that always move to~$v_2$ and~$v'_2$ from~$v_1$, respectively.
Both strategies are optimally resilient.
Hence, the algorithm from Section~\ref{sec_results} may yield either one, depending on the underlying parity game solver used.
Intuitively, however,~$\sigma'$ is preferable for Player~$0$, as a play prefix ending in~$v'_2$ may proceed to her winning region if a single disturbance occurs.
All plays encountering~$v_2$ at some point, however, are losing for her.
Hence, another interesting avenue for further research is to study \openquestion{how to recover from losing}, i.e., how to construct strategies that leverage disturbances in order to leave Player~$1$'s winning region. 
For safety games, this has been addressed by Dallal, Neider, and Tabuada~\cite{DBLP:conf/cdc/DallalNT16}.

\begin{figure}[h]
	\centering
	\begin{tikzpicture}
		
		\begin{scope}[shift={(4.5,-.35)}]
		\ParityVertexZero{0}{v_0}{0}{(0,0)}
		\ParityVertexZero{1}{v_1}{1}{(2,0)}
		\ParityVertexZero{2}{v_2}{1}{(4,.4)}
		\ParityVertexZero{3}{v'_2}{1}{(4,-.4)}
		\ParityVertexZero{4}{v_3}{0}{(6,.4)}
		
		\path
			(0) edge[loop left] (0) edge[fault] (1)
			(1) edge (2) edge (3)
			(2) edge[loop right] (2)
			(3) edge[loop right] (3) edge[fault] (4)
			(4) edge[loop right] (4);

		\begin{pgfonlayer}{background}
			\coordinate (W1-west) at ($(0) ! .5 ! (1)$);
			\coordinate (W1-southwest) at (W1-west |- 3);
			\coordinate (W1-northwest) at (W1-west |- 2);
			\coordinate (W1-southeast) at (3);
			\coordinate (W1-northeast) at (2);
		
			\draw[black!15,thick,rounded corners,fill=black!5]
					($(W1-southwest) - (-.1cm,.45cm)$) --
					($(W1-northwest) + (.1cm,.45cm)$) --
					($(W1-northeast) + (1cm,.45cm)$) --
					($(W1-southeast) + (1cm,-.45cm)$) --
					cycle;
			\node[anchor=north west] at ($(W1-northwest) + (.1cm,.45cm)$) {$\winreg_1$};
			
			\coordinate (W0-northwest) at (0 |- 2);
			\coordinate (W0-northeast) at (4);
			\coordinate (W0-southwest) at (0 |- 3);
			\coordinate (W0-southeast) at (4 |- 3);
			
			\draw[black!15,thick,rounded corners,fill=black!5]
					($(W0-southwest) - (1cm,.75cm)$) --
					($(W0-northwest) + (-1cm,.45cm)$) --
					($(W1-northwest) + (-.1cm,.45cm)$) --
					($(W1-southwest) + (-.1cm,-.55cm)$) --
					($(W1-southeast) + (1.2cm,-.55cm)$) --
					($(W1-northeast) + (1.2cm,.45cm)$) --
					($(W0-northeast) + (1cm,.45cm)$) --
					($(W0-southeast) + (1cm,-.75cm)$) --
					cycle;
			\node[anchor=north west] at ($(W0-northwest) + (-1cm,.45cm)$) {$\winreg_0$};
		\end{pgfonlayer}	
		\end{scope}

	\end{tikzpicture}	
	\caption{Intuitively, moving from~$v_1$ to~$v'_2$ is preferable for Player~$0$, as it allows her to possibly ``recover'' from a first disturbance with the ``help'' of a second one.}
	\label{fig:recovery}
\end{figure}
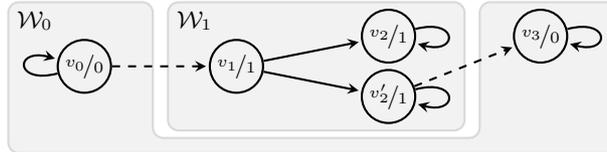

The previous example shows that Player~$0$ can still make ``meaningful'' choices even if the play has moved outside her winning region.
The game~$\game$ shown in Figure~\ref{fig:omega-resilience-distinction} demonstrates that she can do so as well when remaining in vertices of resilience~$\omega$.
Every vertex in~$\game$ has resilience~$\omega$, since every play with finitely many disturbances eventually remains in vertices of color~$0$.
Moreover, the only choice to be made by Player~$0$ is whether to move to vertex~$v_1$ or to vertex~$v'_1$ from vertex~$v_0$.
Let~$\sigma$ and~$\sigma'$ be positional strategies that implement the former and the latter choice, respectively.

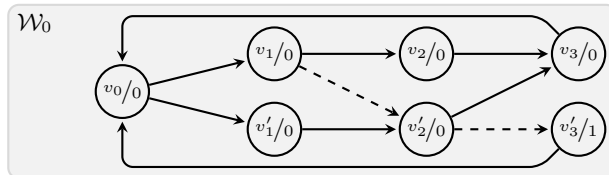
\begin{figure}[h]
	\centering
	\begin{tikzpicture}
		\ParityVertexZero{0}{v_0}{0}{(0,0)}
		\ParityVertexZero{1}{v_1}{0}{(2,.5)}
		\ParityVertexZero{2}{v_2}{0}{(4,.5)}
		\ParityVertexZero{3}{v_3}{0}{(6,.5)}
		\ParityVertexZero{4}{v'_1}{0}{(2,-.5)}
		\ParityVertexZero{5}{v'_2}{0}{(4,-.5)}
		\ParityVertexZero{6}{v'_3}{1}{(6,-.5)}
		
		\path
			(0) edge (1) edge (4)
			(1) edge (2) edge[fault] (5)
			(2) edge (3)
			(4) edge (5)
			(5) edge[fault] (6) edge (3);
			
		\path[draw,rounded corners,->,thick]
			(3) -- ($(3) + (-.5cm,.5cm)$) -| (0);
		\path[draw,rounded corners,->,thick]
			(6) -- ($(6) - (.5cm,.5cm)$) -| (0);
			
		\begin{pgfonlayer}{background}
			\coordinate (northwest) at (0 |- 1);
			\coordinate (southwest) at (0 |- 4);
			\coordinate (northeast) at (3);
			\coordinate (southeast) at (6);
			
			\coordinate (W0-label) at ($(northwest) ! .5 ! (southwest)$);
			
			\draw[black!15,thick,rounded corners,fill=black!5]
				($(southwest) + (-1.5cm,-.65cm)$) --
				($(northwest) + (-1.5cm,.65cm)$) --
				($(northeast) + (.5cm,.65cm)$) --
				($(southeast) + (.5cm,-.65cm)$) --
				cycle;
				
			\node[anchor=north west] at ($(northwest) + (-1.5cm,.65cm)$) {$\winreg_0$};
		\end{pgfonlayer}

	\end{tikzpicture}
	\caption{Moving from~$v_0$ to~$v_1$ allows Player~$0$ to minimize visits to odd colors, while moving to~$v'_1$ allows her to minimize the occurrence of disturbances.}
	\label{fig:omega-resilience-distinction}
\end{figure}

First consider a scenario in which visiting an odd color models the occurrence of some undesirable event, e.g., that a request has not been answered.
In this case, Player~$0$ should aim to prevent visits to~$v'_3$ in~$\game$, the only vertex of odd color.
Hence, the strategy~$\sigma$ should be more desirable for her, as it requires two disturbances in direct succession in order to visit $v'_3$.
When playing consistently with~$\sigma'$, however, a single disturbance suffices to visit~$v'_3$.

On the other hand, consider a setting in which Player~$0$'s goal is to avoid the occurrence of disturbances.
In that case,~$\sigma'$ is preferable over~$\sigma$, as it allows for fewer situations in which disturbances may occur, since no disturbances are possible from vertices~$v_2$ and~$v_3$.

Note that the goals of minimizing visits to vertices of odd color and minimizing the occurrence of disturbances are not contradictory:
if both events are undesirable, it may be optimal for Player~$0$ to combine the strategies~$\sigma$ and~$\sigma'$.
In general, it is interesting to study how to \openquestion{how to best brace for a finite number of disturbances}.

Recall that, due to Theorem~\ref{theorem_main_parity}, optimally resilient strategies for parity games do not require memory.
In contrast, the game shown in Figure~\ref{fig:memory} demonstrates that additional memory can serve to further improve such strategies.
Any strategy for Player~$0$ that does not stay in~$v_1$ from some point onwards is optimally resilient.
However, every visit to~$v_2$ risks a disturbance occurring, which would lead the play into a losing sink for Player~$0$.
Hence, it is in her best interest to remain in vertex~$v_1$ for as often as possible, thus minimizing the possibility for disturbances to occur.
This behavior does, however, require memory to implement, as Player~$0$ needs to count the visits to~$v_1$ in order to not remain in that state ad infinitum. 
Even worse, for each optimally resilient strategy~$\sigma$ with finite memory there exists another optimally resilient strategy that uses more memory, but visits~$v_2$ more rarely than~$\sigma$, reducing the possibilities for disturbances to occur.
Hence, it is interesting to study \openquestion{how to balance avoiding disturbances with satisfying the winning condition.}
This is particularly interesting if there is some cost assigned to disturbances.

\begin{figure}[h]
	\centering
	\begin{tikzpicture}
				
		\begin{scope}[shift={(6.5,.55)}]
		\ParityVertexZero{0}{v_0}{1}{(6,0)}
		\ParityVertexZero{1}{v_2}{2}{(4,0)}
		\ParityVertexZero{2}{v_1}{1}{(2,0)}
		
		\path
			(0) edge[loop right] (0)
			(1) edge[fault] (0)
			(1) edge[bend left=10] (2)
			(2) edge[bend left=10] (1) edge[loop left] (2);
			
		\begin{pgfonlayer}{background}
		
		\draw[black!15,thick,rounded corners,fill=black!5]
			($(0) + (-1cm,.5cm)$) --
			($(0) + (1.6cm,.5cm)$) --
			($(0) + (1.6cm,-.5cm)$) --
			($(0) + (-1cm,-.5cm)$) --
			cycle;
			
		\node[anchor=north west] at ($(0) + (+.9cm,.5cm)$) {$\winreg_1$};
		
		\draw[black!15,thick,rounded corners,fill=black!5]
			($(2) + (-1.5cm,.5cm)$) --
			($(1) + (.9cm,.5cm)$) --
			($(1) + (.9cm,-.5cm)$) --
			($(2) + (-1.5cm,-.5cm)$) --
			cycle;
			
		\node[anchor=north west] at ($(2) + (-1.5cm,.5cm)$) {$\winreg_0$};
			
		\end{pgfonlayer}	
		\end{scope}

	\end{tikzpicture}
	\caption{Additional memory allows Player~$0$ to remain in~$v_1$ longer and longer, thus decreasing the potential for disturbances.}
	\label{fig:memory}
\end{figure}
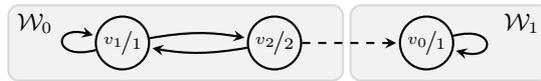

In quantitative games, there is a further tradeoff between resilience and the semantic quality of strategies. As a simple example, consider the parity game in Figure~\ref{fig:tradeoff} and assume, for the sake of argument, that Player~$0$ aims to maximize the maximal color seen infinitely often. Thus, when it comes to semantic quality of strategies, Player~$0$ prefers moving to $v_2$ over moving to $v_0$, when starting in $v_1$. However, $v_2$ has  resilience one while $v_0$ has resilience two. Hence, in this aspect, Player~$0$ prefers moving to $v_0$ over moving to $v_2$. In general, it is an interesting question to \openquestion{determine the tradeoff between resilience and semantic quality and to compute strategies that optimize both aspects, if possible.}

\begin{figure}[h]
	\centering
	\begin{tikzpicture}
				
		\begin{scope}[shift={(6.5,.55)}]
		\ParityVertexZero{1}{v_1}{1}{(0,0)}
		\ParityVertexZero{0}{v_0}{0}{(2,.5)}
		\ParityVertexZero{2}{v_2}{2}{(2,-.5)}
		\ParityVertexZero{3}{v_0'}{0}{(4,.5)}
		\ParityVertexZero{4}{v_3}{3}{(6,0)}

		\path
			(1) edge (0)
			(1) edge (2)
			(2) edge[fault] (4)
			(0) edge[fault] (3)
			(3) edge[fault] (4)
			(4) edge[loop right] ()
			(0) edge[loop above] ()
			(2) edge[loop below] ()
			(3) edge[loop above] ();
			
		\begin{pgfonlayer}{background}
		
		\draw[black!15,thick,rounded corners,fill=black!5]
			($(0) + (-3.1cm,1cm)$) --
			($(0) + (2.9cm,1cm)$) --
			($(0) + (2.9cm,-2cm)$) --
			($(0) + (-3.1cm,-2cm)$) --
			cycle;
			
		\node[anchor=north west] at ($(0) + (-3cm,.95cm)$) {$\winreg_0$};
		
		\draw[black!15,thick,rounded corners,fill=black!5]
			($(4) + (-.9cm,1.5cm)$) --
			($(4) + (.9cm,1.5cm)$) --
			($(4) + (.9cm,-1.5cm)$) --
			($(4) + (-.9cm,-1.5cm)$) --
			cycle;
			
		\node[anchor=north west] at ($(4) + (-.8cm,1.45cm)$) {$\winreg_1$};
			
		\end{pgfonlayer}	
		\end{scope}

	\end{tikzpicture}
	\caption{A tradeoff between resilience and semantic quality (measured in the maximal color occurring infinitely often).}
	\label{fig:tradeoff}
\end{figure}
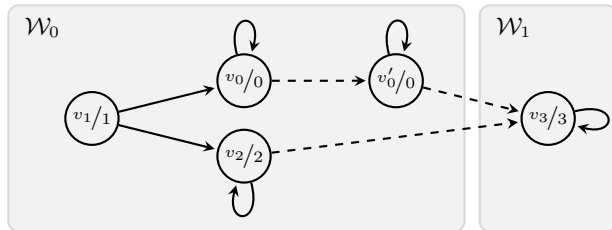

Finally, another important and interesting aspect, which falls outside the scope of this paper, is to provide general guidelines and best practices on how to model synthesis problems by games with disturbances.
We will address these problems in future research.

\section{Related Work}
\label{sec_relatedwork}
The notion of \emph{unmodeled intermittent disturbances} in infinite games has recently been formulated by Dallal, Neider, and Tabuada~\cite{DBLP:conf/cdc/DallalNT16}. In that work, the authors also present an algorithm for computing optimally resilient strategies for safety games with disturbances, which is an extension of the classical attractor computation~\cite{GraedelThomasWilke02}. Due to the relatively simple nature of such games, however, this algorithm cannot easily be extended to handle more expressive winning conditions, and the approach presented in this work relies on fundamentally different ideas.

\emph{Resilience} is not a novel concept in the context of reactive systems synthesis.
It appears, for instance, in the work by Topcu et al.~\cite{DBLP:conf/hybrid/TopcuOLM12} as well as Ehlers and Topcu~\cite{DBLP:conf/hybrid/EhlersT14}. 
A notion of resilience that is very similar to the one considered here has been proposed by Huang et~al.~\cite{DBLP:journals/tse/HuangPSW16}, where the game graph is augmented with so-called ``error edges''. However, this setting differs from the one studied in this work in various aspects.
Firstly, Huang et~al.\ work in the framework of concurrent games and model errors as being under the control of Player~$1$. This contrasts to the setting considered here, in which the players play in alternation and disturbances are seen as rare events rather than antagonistic to Player~$0$.
Secondly, Huang et~al.\ restrict themselves to safety games, whereas we consider a much broader class of infinite games.
Finally, Huang et~al.\ compute resilient strategies with respect to a fixed parameter $k$, thus requiring to repeat the computation for various values of $k$ to find optimally resilient strategies. In contrast, our approach computes an optimal strategy in a single run.
Hence, they consider a more general model of interaction, but only a simple winning condition, while the notion of disturbances considered here is incomparable to theirs.

Related to resilience are various notions of \emph{fault tolerance}~\cite{DBLP:journals/toplas/AttieAE04,DBLP:conf/fsttcs/BrihayeGHMPR15,DBLP:journals/sttt/EbnenasirKA08,DBLP:journals/fmsd/GiraultR09} and \emph{robustness}~\cite{DBLP:journals/acta/BloemCGHHJKK14,DBLP:conf/cav/BloemCHJ09,DBLP:journals/corr/BloemEJK14,DBLP:journals/jcss/BloemJPPS12,DBLP:journals/tecs/MajumdarRT13,DBLP:journals/tac/TabuadaCRM14,DBLP:conf/csl/TabuadaN16}.
For instance, Brihaye et al.~\cite{DBLP:conf/fsttcs/BrihayeGHMPR15} consider quantitative games under failures, which are a generalization of sabotage games~\cite{DBLP:conf/birthday/Benthem05}.
The main difference to our setting is that Brihaye et al.\ consider failures---embodied by a saboteur player---as antagonistic, whereas we consider disturbances as non-antagonistic events. Moreover, solving a parity game while maintaining a cost associated with the sabotage semantics below a given threshold is \exptime-complete, whereas our approach computes optimally resilient controllers for parity conditions in quasipolynomial time.

Besides fault tolerance, robustness in the area of reactive controller synthesis has also attracted considerable interest in the recent years, typically in settings with specifications of the form~$\varphi \Rightarrow \psi$ stating that the controller needs to fulfill the guarantee~$\psi$ if the environment satisfies the assumption~$\varphi$. 
A prominent example of such work is that of Bloem et al.~\cite{DBLP:journals/acta/BloemCGHHJKK14}, in which the authors understand robustness as the property that ``if assumptions are violated temporarily, the system is required to recover to normal operation with as few errors as possible'' and consider the synthesis of robust controllers for the GR(1) fragment of Linear Temporal Logic~\cite{DBLP:journals/jcss/BloemJPPS12}. Other examples include quantitative synthesis~\cite{DBLP:conf/cav/BloemCHJ09}, where robustness is defined in terms of payoffs, and the synthesis of robust controllers for cyber-physical systems~\cite{DBLP:journals/tecs/MajumdarRT13,DBLP:journals/tac/TabuadaCRM14}.
For a more in-depth discussion of related notions of resilience and robustness in reactive synthesis, we refer the interested reader to Dallal, Neider, and Tabuada's section on related work~\cite[Section I]{DBLP:conf/cdc/DallalNT16}. Moreover, a survey of a large body of work dealing with robustness in reactive synthesis has been presented by Bloem et al.~\cite{DBLP:journals/corr/BloemEJK14}.

Finally, note that for the special case of parity games, we can also characterize vertices of finite resilience (cf.~ Subsection~\ref{subsec_finiteresil}) by a reduction to finding optimal strategies in energy parity games~\cite{DBLP:journals/tcs/ChatterjeeD12}, which yields the same complexity as our algorithm (though such a reduction would not distinguish between vertices with resilience~$\omega$ and vertices with resilience~$\omega+1$. Also, it is unclear if and how this reduction can be extended to other winning conditions and if custom-made solutions would be required for each new class of game. By contrast, our refinement-based approach works for any class of infinite games that satisfies the mild assumptions discussed in Section~\ref{sec_discussion}.

\section{Conclusion} 
\label{sec_conc}

We presented an algorithm for computing optimally resilient strategies in games with disturbances that is applicable to any game that satisfies some mild (and necessary) assumptions. Thereby, we have vastly generalized the work of Dallal, Neider, and Tabuada, who only considered safety games. Furthermore, we showed that optimally resilient strategies are typically of the same size as classical winning strategies. Finally, we have illustrated numerous novel phenomena that appear in the setting with disturbances but not in the classical one. 
Studying these phenomena is a very promising direction of future work.

As part of future work, we are currently implementing our proposed method on top of the parity game solver Oink~\cite{DBLP:conf/tacas/Dijk18} and SCOTS~\cite{DBLP:conf/hybrid/RunggerZ16}, a tool for the synthesis of controllers in the context of dynamic and cyber-physical systems.
Besides developing an end-to-end synthesis tool for controllers of dynamic and cyber-physical systems, a major part of this effort is to evaluate the impact of the polynomial overhead as compared to classical parity game solvers.
Preliminary experiments with this prototype implementation suggest that the additional overhead does not impact the overall performance much.

\bibliographystyle{plain}      
\bibliography{biblio}

\end{document}